\documentclass[11pt]{article}
\usepackage{natbib}
\usepackage{float} 
\usepackage{graphicx}

\usepackage{marsden_article}
\usepackage{amssymb}

\usepackage{color}
\usepackage{enumerate}
\usepackage{tikz-cd}
\usetikzlibrary{decorations.pathmorphing}
\usepackage[siunitx]{circuitikz}
\usepackage{verbatim}

\usepackage{bold-extra}

\newtheoremstyle{obs}
  {3pt}
  {3pt}
  {}
  {}
  {\bfseries}
  {.}
  {.5em}
  {}

\theoremstyle{obs}

\newtheorem{remark}[theorem]{Remark}
\newtheorem{example}[theorem]{Example}

\newcommand{\F}{\mathcal{F}}
\newcommand{\B}{\mathcal{B}}
\newcommand{\R}{\mathbb{R}}

\newcommand{\D}{\mathcal{D}}

\DeclareMathOperator{\Dir}{\mathfrak{Dir}}

\newcommand{\st}{\;\ifnum\currentgrouptype=16 \middle\fi|\;}

\begin{document}
\title{New insights in the geometry and interconnection of port-Hamiltonian systems}

\author{M.\ Barbero-Li\~n\'an\textsuperscript{\textsection,$\ddagger$}, H.\ Cendra\textsuperscript{$\dagger$}, E.\ Garc\'{\i}a-Tora\~{n}o Andr\'{e}s\textsuperscript{$\dagger$},  and \\ D.\ Mart\'{\i}n de Diego\textsuperscript{$\ddagger$}
\\[2mm]
{\small \textsection\, Departamento de Matem\'atica Aplicada, Universidad Polit\'ecnica de Madrid,} \\
{\small Av. Juan de Herrera 4, 28040 Madrid, Spain}
\\[2mm]
{\small $\dagger$ Departamento de Matem\'atica,
Universidad Nacional del Sur, CONICET,}  \\
{\small  Av.\ Alem 1253, 8000 Bah\'ia Blanca, Argentina}
\\[2mm]
{\small $\ddagger$  Instituto de Ciencias Matem\'aticas (CSIC-UAM-UC3M-UCM),} \\
{\small C/Nicol\'as Cabrera 13-15, 28049 Madrid, Spain}}

\date{}

\maketitle

\begin{abstract}
We discuss a new geometric construction of port-Hamiltonian systems. Using this framework, we revisit the notion of interconnection providing it with an intrinsic description. Special emphasis on theoretical and applied examples is given throughout the paper to show the applicability and the novel contributions of the proposed framework.
\end{abstract}

\section{Introduction}\label{sec:Intro}

Dirac structures were introduced in~\citep{CouWei}, partially motivated by the Dirac theory of constraints,  as a unified approach to symplectic and Poisson geometry (see~\citep{Courant,Dorfman} for more details). Besides their genuine geometric interest, Dirac structures have proven to be extremely useful in the modeling of physical systems. Much of its interest comes from the observation that Dirac structures allow for a definition of implicit Hamiltonian systems (as in~\citep{VdSM2,DVdS}) which are general enough to encompass many dynamical systems of interest in mathematical physics. Based on this observation, A.\ van der Schaft and B.\ Maschke defined the notion of port-Hamiltonian system (meaning a Hamiltonian systems with ``ports'') which describes general Hamiltonian systems that can be interconnected through their ports to build more complex physical systems. The approach of port-Hamiltonian systems has also been successfully employed to describe irreversible thermodynamic processes~\citep{Thermo}.

As recognized by A.\ van der Schaft and collaborators (B.\ Maschke, M.\ Dalsmo, and many others), the notion of port-Hamiltonian system unifies geometric mechanics with network theory, and therefore provides a natural framework to study interconnection. From this perspective, a good mathematical description of the interconnection will result in a good model of the system under study (this is especially important from the standpoint of numerical analysis). The key observation that Dirac structures naturally encode energy-preserving connections among subsystems can be found in~\citep{VdSM2,DVdS}, which also contain some earlier references and examples supporting this claim. We refer the reader to~\citep{VdS-Book} for a recent survey on the different aspects of port-Hamiltonian systems. In~\cite{vdSMaschkeDistributed}, the so-called Dirac--Stokes structures are used to describe certain  partial differential equations of interest in physics as port-Hamiltonian systems. A spatial discretization of those dynamical systems has been studied in~\cite{DiscreteBoundary}.

In this paper, we wish to take an alternative look to the theory of port-Hamiltonian systems which emphasizes their intrinsic description. The main contributions of this work are:
\begin{enumerate}[1)]
\item We review the theory of port-Hamiltonian systems from a geometric perspective. We show that the composition of port-Hamiltonian systems can be intrinsically defined by means of the forward and backward operators applied to Dirac structures.
\item We propose an alternative framework for a class of port-Hamiltonian systems based on the notion of coisotropic structure. Within this formalism, we give an interconnection procedure which includes some well-known interconnection schemes (for instance, the composition of port-Hamiltonian systems or the use of the tensor product to interconnect Dirac systems). Some illustrative examples such as  electric circuits or multi-body mechanical systems are discussed in detail.
\item We highlight that many constructions in the literature can be seen as dual to each other via the forward and backward operations. In particular, the so-called ``Representation II'' and ``Representation III'' in~\citep{DVdS} are recovered in our approach through the forward and backward constructions.
\end{enumerate}

The paper is structured as follows. In Section~\ref{Sec:Background} we provide some background to make the text reasonably self-contained. This includes the essential properties of Dirac structures, the notion of forward and backward of a Dirac structure, and some basic theoretical examples. Section~\ref{Sec:PHS} discusses the standard framework of port-Hamiltonian systems. We present a new description of the composition of port-Hamiltonian systems in terms of the forward and backward of Dirac structures, and check that it coincides with the definition due to J.\ Cervera, A.\ J.\ van der Schaft, and A.\ Ba{\~n}os. In Section~\ref{sec:InputOutput} we introduce a new method to construct and interconnect port-Hamiltonian systems that guarantees without additional proofs that the geometric structures and properties associated with the initial systems are preserved after interconnection. The method is illustrated by many examples which help to compare the proposed framework with the existing literature. The final sections are devoted to future work and two technical appendices.

\paragraph{Conventions and terminology.} All objects in this paper are assumed to be smooth, unless otherwise stated. The maps between vector spaces and vector bundles are assumed to be linear. If $\pi_E\colon E\to M$ is a vector bundle, a distribution $D$ is an assignment of a subspace $x\mapsto D_x\subset E_x$ for each $x\in M$. We will say that the distribution $D$ is regular if $\text{dim}(D_x)$ is independent of $x$. Regular (smooth) distributions are called subbundles. In general, our conventions agree with those on~\citep{AbMa}.

\paragraph{Acknowledgements.} The authors have been partially supported by Ministerio de Econom\'ia y Competitividad (MINECO, Spain) under grant MTM 2015-64166-C2-2P. MB and DMdD acknowledge financial support from the Spanish Ministry of Economy and Competitiveness, through the   research grants MTM2013-42870-P, MTM2016-76702-P and ``Severo Ochoa Programme for Centres of Excellence'' in R\&D (SEV-2015-0554). HC thanks the following institutions from Argentina: CONICET, Project PIP 2013-2015: 11220120100532CO; ANPCyT Project PICT 2013 1302; Univesidad Nacional del Sur, PGI 24L098. He also thanks the hospitality of ICMAT and also, the University
of Gr\" oningen, and many useful conversations with Arjan van der Schaft. EGTA thanks the CONICET for financial support through a Postdoctoral Grant.

\section{Dirac structures: forward and backward}
\label{Sec:Background}

In this section we review the essential definitions and results we need from Dirac structures and the dynamical systems associated to them.

\subsection{Linear Dirac structures}

Let $V$ be a $n$-dimensional vector space and $V^*$ be its dual space. We consider the non-degenerate symmetric pairing $\ll \cdot, \cdot \gg $ on $V\oplus V^*$ given by
\[
\ll (v_1,\alpha_1),(v_2,\alpha_2) \gg=\langle \alpha_1, v_2\rangle + \langle \alpha_2,v_1\rangle\,,
\label{Eq:Df_Sym_Pair}
\]
for $ (v_1,\alpha_1),(v_2,\alpha_2)\in V\oplus V^*$, where $\langle \cdot, \cdot \rangle$ is the natural pairing between $V^*$ and $V$. A \emph{linear Dirac structure on $V$} is a subspace $D\subset V\oplus V^*$ such that $D=D^\perp$, where $D^\perp$ is the orthogonal subspace of $D$ relative to the pairing $\ll \cdot, \cdot \gg$.  Note that according to this definition, the condition $D=D^\perp$ implies that $\langle \alpha, v \rangle=0$ for each $(v,\alpha)\in D$. Actually, it is not hard to check that a vector subspace $D\subset V\oplus V^*$ is a Dirac structure on $V$ if and only if it is maximally isotropic with respect to the symmetric pairing $\ll \cdot, \cdot \gg $, namely, if ${\rm dim} \, D=n$ and $\ll (v_1,\alpha_1), (v_2,\alpha_2) \gg =0$ for all $(v_1,\alpha_1)$, $ (v_2,\alpha_2)$ in $D$. Using this, we have the following three important examples of Dirac structures:

\begin{enumerate}[a)]
\item Let $F$ be a subspace of $V$, the annihilator $F^\circ$ of $F$ is the subspace of $V^*$ defined as follows
\begin{equation*}
F^\circ=\{\alpha \in V^* \st \langle \alpha, v \rangle=0 \; \mbox{ for all }\; v\in F\}.
\end{equation*}
It can be easily proved that $D_F=F\oplus F^\circ$ is a Dirac structure on $V$.
\item On a presymplectic vector space $(V,\omega)$, the graph of the musical isomorphism $\omega^\flat$ defines a Dirac structure that we denote $D_\omega$:
\begin{equation*}
D_\omega=\{ (v,\alpha)\in V \oplus V^* \st \alpha=\omega^\flat(v)\},
\end{equation*}
(recall that $\omega^\flat\colon V\rightarrow V^*$ is defined by $\omega^\flat(u)(v)=\omega(u,v)$ for all $u$, $v$ in $V$).
\item Let $\Lambda:V^*\times V^*\to \R$ is a bivector on $V$. Then $\sharp_\Lambda:V^*\to V$ is defined as $\langle\beta,\sharp_\Lambda(\alpha)\rangle=\Lambda(\beta,\alpha)$, with $\alpha,\beta\in V^*$, and its graph defines the Dirac structure
\begin{equation*}
D_\Lambda=\{(v,\alpha)\in V \oplus V^* \st v=\sharp_\Lambda(\alpha)\}.
\end{equation*}
\end{enumerate}

The following fundamental result can be found in~\citep{Courant}:
\begin{proposition}
Let $D$ be a Dirac structure on $V$. Define the subspace $F_D\subset V$ to be the projection of $D$ on $V$. Let $\omega_D$ be the 2-form on $F_D$ given by $\omega_D (u,v)=\alpha(v)$, where $u\oplus \alpha\in D$. Then $\omega_D$ is a skew form on $F_D$. Conversely, given a vector space $V$, a subspace $F\subset V$ and a skew form $\omega$ on $F$,
\begin{equation*}
D_{F,\omega}=\{u\oplus \alpha \st u\in F, \; \alpha(v)=\omega(u,v) \; \mbox{for all} \; v\in F\}
\end{equation*}
is the only Dirac structure $D$ on $V$ such that $F_D=F$ and $\omega_D=\omega$. \label{Prop:SkewFormDirac}
\end{proposition}
\noindent In other words, a Dirac structure $D$ on $V$ is uniquely determined by a subspace $F_D\subset V$ and a 2-form $\omega_D$. The case $F=V$ is the example (b) above. The set of Dirac structures on $V$ will be denoted by $\Dir(V)$.

One of the remarkable properties of Dirac structures is that there are ``push-forward'' and ``pull-back'' operations. We now discuss these constructions, which will play a major role in the text. Let $\varphi\colon V\to  W$ be a linear map between vector spaces, and let $D_W$ be a Dirac structure on $W$. It is possible to induce a Dirac structure $D_V$ on $V$, the \emph{backward of $D_W$ by $\varphi$}, denoted by $D_V=\B \varphi(D_W)$, as follows:
\[
D_V=\B \varphi(D_W)=\{(v,\varphi^*w^*)\in V\oplus V^* \st v\in V, \; w^*\in W^*, \; (\varphi v,w^*)\in D_W\}.
\]
In a similar way, if $D_V$ is a Dirac structure on $V$, we can construct a Dirac structure $D_W=\F \varphi(D_V)$ on $W$, the \emph{forward of $D_V$ by $\varphi$},
\[
D_W=\F \varphi(D_V)=\{(\varphi v,w^*)\in W\oplus W^* \st v\in V, \; w^*\in W^*, \; (v,\varphi^*w^*)\in D_V\}.
\]
One can think of the assignments $\B \varphi$ and $\F \varphi$ as maps between the sets $\Dir(V)$ and $\Dir(W)$ (see the diagram below). The following rules hold:
\[
\F (\varphi_1\circ \varphi_2)=\F \varphi_1\circ \F \varphi_2,\qquad \B (\varphi_1\circ \varphi_2)=\B \varphi_2\circ \B \varphi_1.
\]
A more detailed exposition of these notions might be found in e.g.~\citep{Burs,RadkoBursztyn,2012CenRaYo}.

\begin{figure}[h!]
\centering
\includegraphics{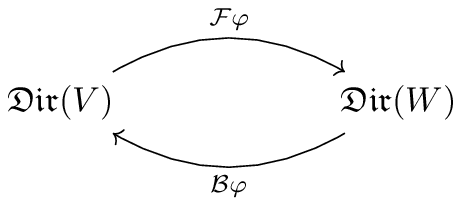}
\end{figure}

\paragraph{Isotropic and coisotropic structures.} Besides the notion of Dirac structure, we will need more general structures on $V\oplus V^*$. A subspace $\Sigma\subset V\oplus V^*$ is called:
\begin{enumerate}[1)]
\item  \textit{isotropic} if $\Sigma \subset \Sigma^\perp$.
\item  \textit{coisotropic} if $ \Sigma^\perp \subset \Sigma$.
\end{enumerate}
Recall that ``$\perp$'' denotes the orthogonal subspace w.r.t. the pairing $\ll \cdot, \cdot \gg $. It follows that if $\Sigma$ is isotropic then $\Sigma^\perp$ is coisotropic and conversely. We will also say that $\Sigma$ is an isotropic -or coisotropic- structure. It can be shown that a Dirac structure $D\subset V\oplus V^*$ is characterized by one of the following three equivalent conditions: $D$ is maximally isotropic, $D$ is minimally coisotropic or $D$ is both isotropic and coisotropic. One can also prove that a subspace $\Sigma$ is isotropic if, and only if, it is a subspace of a Dirac structure. Similarly, $\Sigma$ is coisotropic if, and only if, it contains a Dirac structure. We will need later the following lemma:

\begin{lemma}\label{lem:inclussions} Let $V,W$ be vector spaces and $\varphi\colon V\to W$ a linear map.
\begin{enumerate}[i)]
\item  If $\Sigma\subset V\oplus V^*$ is isotropic (resp. coisotropic), then $\F\varphi(\Sigma)\subset W\oplus W^*$ is isotropic (resp. coisotropic).
\item  If $\Sigma'\subset W\oplus W^*$ is isotropic (resp. coisotropic), then $\B\varphi(\Sigma')\subset V\oplus V^*$ is isotropic (resp. coisotropic).
\end{enumerate}
\end{lemma}
\begin{proof} Let $\Sigma$ be isotropic. Consider a Dirac structure $D$ such that $\Sigma\subset D$. Since the operator $\F\varphi$ preserves the inclusion, $\F\varphi(\Sigma)\subset \F\varphi(D)$. As $\F\varphi(D)$ is a Dirac structure, $\F\varphi(\Sigma)$ is isotropic. If $\Sigma$ is coisotropic, the reasoning is similar. The backward case is analogous using that $\B\varphi$ preserves the inclusion.
\end{proof}

A more detailed description of isotropic and coisotropic structures in vector spaces can be found in Appendix~\ref{ap:A}.

\subsection{Dirac structures on a manifold}\label{sec:DiracOnManifolds}

A \emph{Dirac structure} $D$ on a manifold $M$, is a vector subbundle of the Whitney sum $TM \oplus T^*M$ such that $D_x\subset T_xM \oplus T_x^*M$ is a linear Dirac structure on the vector space $T_xM$ at each point $x\in M$. A \textit{Dirac manifold} is a manifold $M$ with a Dirac structure $D$ on $M$.

From Proposition~\ref{Prop:SkewFormDirac}, a Dirac structure on $M$ yields a distribution $F_{D_x}\subset T_xM$ whose dimension is not necessarily constant, carrying a 2-form $\omega_D(x)\colon F_{D_x}\times F_{D_x} \rightarrow \mathbb{R}$ for all $x\in M$. The following result is proved in~\citep{DVdS}, (see also~\citep{2006YoshiMarsdenI}).
\begin{theorem}\label{Thm:DiracStructM}
Let $M$ be a manifold,  $\omega$ be a 2-form on $M$ and $F$ be a regular distribution on $M$. Define the skew-symmetric bilinear form $\omega_F$ on $F$ by restricting $\omega$ to $F\times F$. For each $x\in M$, define
\begin{align*}
D_{\omega_F }(x)=\left\{(v_x,\alpha_x)\in T_xM \oplus T_x^*M \right. \st& v_x\in F_x, \; \alpha_x(u_x)=\omega_F(x)(v_x,u_x) \; \\
& \left. \mbox{for all } \; u_x\in F_x\right\}.
\end{align*}
Then $D_{\omega_F}\subset TM \oplus T^*M$ is a Dirac structure on $M$. In fact, it is the only Dirac structure $D$ on $M$ satisfying $F_x=F_{D_x}$ and $\omega_F(x)=\omega_D(x)$ for all $x\in M$.
\end{theorem}

\noindent As usual, we have used the terminology \emph{regular distribution} to mean that $F$ has constant rank.  Examples of Theorem~\ref{Thm:DiracStructM} are the case $\omega=0$ where $D_{\omega_F}=F\oplus F^\circ\subset TM\oplus T^*M$, and the case $F=TM$ where $D_\omega$ is the graph of $\omega$.

The dual version of Theorem~\ref{Thm:DiracStructM} is as follows (see, for instance, \citep{DVdS}).

\begin{theorem}\label{Thm:DiracStructMDual} Let $M$ be a manifold and let $B\colon T^*M\times T^*M\rightarrow \mathbb{R}$ be a skew-symmetric two-tensor. Given a regular codistribution $F^{(*)}\subset T^*M$ on $M$, define the skew-symmetric two-tensor $B_{F^{(*)}}$ on $F^{(*)}$ by restricting $B$ to $F^{(*)}\times F^{(*)}$. For each $x\in M$, let
\begin{align*}
D_{B_{F^{(*)}}}(x)=\left\{(v_x,\alpha_x)\in T_xM\times T_x^*M \right. &\st \alpha_x\in F^{(*)}_x,\; \beta_x(v_x) = B_{F^{(*)}}(x)(\beta_x,\alpha_x) \; \\
& \left. \mbox{for all } \; \beta_x\in F^{(*)}_x\right\}.
\end{align*}
Then $D_{B_{F^{(*)}}}\subset TM \oplus T^*M$ is a Dirac structure on $M$.
\end{theorem}

\noindent As an example, let $(M,\Lambda)$ be a Poisson manifold where $B=\Lambda\colon T^*M \times T^*M \rightarrow \mathbb{R}$. If $F^{(*)}=T^*M$, then the Dirac structure defined in Theorem~\ref{Thm:DiracStructMDual} is the graph of the Poisson structure thought of as a map from $T^*M$ to $TM$.

\remark It is proved in \citep{DVdS} that under the assumptions of constant rank of the distributions/codistributions (viz. regularity), Theorems~\ref{Thm:DiracStructM} and~\ref{Thm:DiracStructMDual} describe the only two canonical Dirac structures that can be defined on a manifold, one associated with a presymplectic form and the other one with a two-tensor.

\begin{remark} A Dirac structure $D$ on $M$ is called \textit{integrable} (see~\citep{Courant}) if the condition
\[
\langle {\rm L}_{X_1} \alpha_2, X_3 \rangle + \langle {\rm L}_{X_2} \alpha_3, X_1 \rangle +\langle {\rm L}_{X_3} \alpha_1, X_2 \rangle=0
\]
is satisfied for all pairs of vector fields and 1-forms $(X_1,\alpha_1)$, $(X_2,\alpha_2)$, $(X_3,\alpha_3)$ in $D$, where ${\rm L}_X$ denotes the Lie derivative along the vector field $X$ on $M$. This condition is linked to the notion of closedness for presymplectic forms and Jacobi identity for brackets, and it is sometimes included in the definition of Dirac structure. The integrability condition is too restrictive to describe, for instance, nonholonomic systems, and for this reason we don't include the integrability in the general definition of a Dirac structure. This is the same convention as in~\citep{VdS-Book} or~\citep{JacYo}.

\end{remark}

Let us assume that we have a smooth map $f\colon M\to N$ between two manifolds $M$ and $N$, and that $D_{TN}\subset TN\oplus T^*N$ is a Dirac structure. At each point $x\in M$, one can use the backward of the map $T_xf$ to construct a subspace of $T_xM\oplus T_x^*M$. This construction defines a distribution of $TM\oplus T^*M$ which in general is not smooth. Whenever it is smooth, it defines a new Dirac structure $D_{TM}$ on $M$ called the \emph{backward} of $D_{TN}$ by the map $Tf$, and we write $D_{TM}=\B (Tf)\left(D_{TN}\right)$. We will use the following results (see e.g.~\citep{Burs,2012CenRaYo}):
\begin{enumerate}[(i)]
\item If $T_xf$ is surjective for each $x\in M$, then $D_{TM}=\B (Tf)\left(D_{TN}\right)$ is a Dirac structure on $M$.
\item If $i_M\colon M\hookrightarrow N$ is a submanifold, then $D_{TM}=\B (Ti_M)\left(D_{TN}\right)=\{(v,\alpha)\in D_{TN}\st \alpha\in TM^\circ\}$ is a Dirac structure if $D_{TN}\cap (\{0\}\oplus TM^\circ)$ has constant rank (the clean-intersection condition), where $TM^\circ=\{\alpha\in T^*N\st \alpha(v)=0,\, \forall v\in TM\}$ is the annihilator.
\item If $D_{TN}$ is given by the graph of a 2-form $\omega$ on $N$, then $D_{TM}=\B (Tf)\left(D_{TN}\right)$ is a Dirac structure on $M$.
\end{enumerate}

\noindent Let now $f\colon M\to N$ be a smooth map, and $D_{TM}$ be a Dirac structure on $M$. When we aim at defining the \emph{forward} of $D_{TM}$, we first need to ask for $Tf$-invariance of $D_{TM}$, meaning that
\[
\F (T_xf)(D_{TM}(x))=\F (T_{x'}f)(D_{TM}(x')),\qquad \text{whenever } f(x)=f(x').
\]
A sufficient condition to ensure that $\F (Tf)(D_{TM})$ defines a Dirac structure is the following:
\begin{enumerate}[(iv)]
\item Let $f\colon M\to N$ be a surjective submersion and $D_{TM}$ be a Dirac structure on $M$. If $D_{TM}$ is $Tf$-invariant and $\{(v,\alpha)\in D_{TM}\st v\in\ker(Tf)\}$ has constant rank, then $D_{TN}=\F (Tf)\left(D_{TM}\right)$ defines a forward Dirac structure.
\end{enumerate}

\remark Note that a Dirac structure on a manifold $M$ will be denoted by $D_{TM}\subset TM\oplus T^*M$ since we will be working with Dirac structures on more general vector bundles. Actually, we will use the expression ``Dirac structure on (the vector bundle) $TM$'' rather than ``Dirac structure on $M$''. For the same reason, given a map $f\colon M\to N$, the usual notations in the literature for the backward and forward are $\B(f)$ and $\F(f)$, respectively, but we will use the notation $\B(Tf)$ and $\F(Tf)$.

\subsection{Dirac structures on vector bundles}\label{sec:DiracOnVBundles}

The definitions of Dirac structure and the forward and backward operations can be extended naturally to vector bundles. We just give the definitions here, and refer to~\citep{2012CenRaYo} for further details.

A Dirac structure on a vector bundle $\pi_{(V,M)}\colon V\to M$ is a vector subbundle $D_V\subset V\oplus V^*$ such that, at each point $x\in M$, $(D_V)_x\subset V_x\oplus V^*_x$ is a linear Dirac structure on the vector space $V_x$. A vector bundle $\pi_{(V,M)}\colon V\to M$ endowed with a Dirac structure will be referred to as a \emph{Dirac vector bundle}, and denoted by $(\pi_{(V,M)},D_V)$.  Note that this definition includes the case of the tangent bundle $TM\to M$ discussed in the previous section. We will use the following notations: if $D\subset V\oplus V^*$ is a Dirac structure, $F_D$ and $F_D^{(*)}$ denote the projections of $D$ on $V$ and $V^*$, respectively.

If $(\pi_{(V,M)},D_V)$ and $(\pi_{(W,N)},D_W)$ are Dirac vector bundles and $\Phi\colon V\to W$ is a vector bundle map, it is possible to define fiberwise the forward and backward Dirac structures $\F(\Phi)(D_V)\subset W\oplus W^*$ and $\B(\Phi)(D_W)\subset V\oplus V^*$.

\vspace{.2cm}
\noindent \underline{Assumption 1}: Whenever we write expressions such as $\F(\Phi)(D_V)$ or $\B(\Phi)(D_W)$, it is assumed that they are well defined vector subbundles, unless otherwise stated.
\vspace{.2cm}

Note that for the case of Dirac structures on a manifold $M$ (that we can now interpret as Dirac structures on the vector bundle $TM\to M$) some necessary conditions to obtain forward and backward Dirac structures have already been reviewed in Section~\ref{sec:DiracOnManifolds}. General conditions for the existence of the forward and backward maps of Dirac vector bundles can be found in~\citep{2012CenRaYo}. Finally, we point out that the concepts of isotropic and coisotropic subspaces also extend to vector bundles without further difficulty.

\subsection{Dirac systems}\label{sec:DiracSystems}

Assume that the vector bundle $TM$ is endowed with a Dirac structure $D\subset TM\oplus T^*M$. In the presence of an \emph{energy} function $E\colon M\to \mathbb{R}$, we will consider the following implicit dynamical system: for a curve $\gamma\colon I\to M$ (where $I\subset\mathbb{R}$ is an interval), we say that $\gamma$ is a solution of the \emph{Dirac system $(D,{\rm d}E)$} if
\begin{equation}\label{eq:DiracDynamics}
\dot \gamma(t) \oplus {\rm d}E\left(\gamma(t)\right)\in D_{\gamma(t)}\quad  \mbox{ for all }\, t\in I.
\end{equation}
The system described by~\eqref{eq:DiracDynamics}, which we call \emph{Dirac system} on $TM$, is general enough to encompass a number of situations of interest in mathematical physics including  of course classical Lagrangian and Hamiltonian systems, but also nonholomic mechanics or electric LC circuits. The system~\eqref{eq:DiracDynamics} is also referred to as an \emph{implicit Hamiltonian system}, see~\citep{VdSM2,DVdS}).

\begin{example} Let $(M,\Lambda)$ be a Poisson manifold, and $H\colon M\to\R$ the Hamiltonian. For the Dirac structure induced by the graph of $\Lambda$ (as in Theorem~\ref{Thm:DiracStructMDual}) the system~\eqref{eq:DiracDynamics} leads to the Hamilton equations
\[
\dot x=\{x,H\}.
\]
\end{example}
\vspace{.3cm}

\begin{example} LC-circuits can be written as a Dirac system on $TM$ with $M=TQ\oplus T^*Q$ with $Q$ being the ``charge space'', whose points represent charges in the different branches of the circuit. This will be discussed later in some detail, see~Example~\ref{ex:LCcir0}.
\end{example}
\vspace{.3cm}

\remark It is possible to extend the definition of Dirac system to include more general Lagrangian submanifolds of $T^*M$ than the graph of an energy function. We refer the reader to~\citep{2017_DiracMorse} for details and examples which show that this broader definition permits to describe in a natural and unified way many dynamical systems of interest in geometric mechanics.

\section{Dynamics of port-Hamiltonian systems}\label{Sec:PHS}

We will now briefly discuss an important family of systems that are widely known as port-Hamiltonian systems. A certain class of them, the so-called Input-Output systems, will be the main interest of this paper. As mentioned in the Introduction, this notion is essentially due to van der Schaft and collaborators (see e.g.~\citep{VdS-Book} and references therein).

\subsection{General definitions}

We first give the basic definitions to study the geometry and the dynamics of systems with ports. The following definitions are similar to those found in e.g.~\citep{DVdS} (see also~\citep{Merker}):
\begin{definition}\normalfont A \emph{port-Hamiltonian structure} (\textsc{ph}-structure) is a triple
\[
A=(\pi_{(U_1,M)},\pi_{(U_2,M)},D_{(U_1\oplus U_2)})
\]
where $\pi_{(U_i,M)}\colon U_i\to M,\,\,i=1,2,$ are vector bundles over $M$ and $D_{U_1\oplus U_2}$ is a Dirac structure on the vector bundle \[
U_1\oplus U_2=\{(u_1,u_2)\in U_1\oplus U_2\st \pi_{(U_1,M)}(u_1)=\pi_{(U_1,M)}(u_2)\}.
\]
\end{definition}
Note that in the case of vector spaces, the Whitney sum $U_1\oplus U_2$ is identified with $U_1\times U_2$.

\begin{definition}\label{def:PH-System}\normalfont A \emph{port-Hamiltonian system} (\textsc{ph}-system) is a pair $(A,{\rm d}E)$ where $A$ is a \textsc{ph}-structure of the form
\[
A=(\tau_M,\pi_{(U_2,M)},D_{(TM\oplus U_2)})
\]
with $\tau_M:TM\to M$ the tangent bundle of $M$, and $E:M\to \R$ is an energy function.
\end{definition}

In other words, a \textsc{ph}-system is given by a \textsc{ph}-structure with $U_1=TM$ and an energy on $M$. For the rest of the paper we will mostly be interested in this case, namely the case where $U_1=TM$. We will however keep the notation $U_1$ because it is convenient, and also because it points towards generalizations. We will say that $U_2$ is the \emph{flow space} and that its dual $U_2^*$ is the \emph{effort space}.

Given a \textsc{ph}-system $(A,{\rm d}E)$, it determines a dynamics as follows:
\begin{equation}\label{eq:PH-System}
(\dot x,u_2,{\rm d}E(x),-\alpha_2)\in D_{(TM\oplus U_2)}.
\end{equation}
Note that~\eqref{eq:PH-System} is a natural extension of the notion of Dirac system~\eqref{eq:DiracDynamics} (which is the case where $U_2=0$). We will also call~\eqref{eq:PH-System} a \textsc{ph}-system.

The coordinates $x$ (the base point) are known as \emph{state} (sometimes also \emph{energy} or \emph{energy-storing}) variables, while $u_2$ and $\alpha_2$  the (external) \emph{flows} and \emph{efforts}, respectively. The choice of the \emph{ports} (i.e. pairs flow-effort) is not unique and depends on the physical system under study, see~\citep{VdS-Book}. Roughly speaking, the variables $x$ correspond to the state variables of the system under study, while the flows and efforts typically model its interaction with other systems.

An important property of \textsc{ph}-systems is the following power balance equation (which holds in view of the isotropy of $D_{(U_1\oplus U_2)}$):
\begin{equation}\label{eq:powerbalance}
\frac{{\rm d}E}{dt}=\langle \alpha_2,u_2\rangle
\end{equation}
From a physical point of view, it expresses that the gain of energy corresponds to the power transmitted by the ports. The sign convention of $\alpha_2$ in~\eqref{eq:PH-System} is the same as in~\citep{VdSM2}, and means that incoming power is counted positively.

\subsection{Open and closed systems}\label{sec:OpenClose}

The power balance equation~\eqref{eq:powerbalance} shows that \textsc{ph}-systems are suitable for describing the dynamics of system with ``open'' ports, allowing for an energy gain or loss of the system under study. The passage from an open system to a closed system (a system with ``closed'' or ``interconnected'' ports) is achieved through the choice of an \emph{interconnecting} Dirac structure $D_I\subset U_2\oplus U_2^*$ which, roughly speaking, routes the power from the various subsystems of the \textsc{ph}-system under study. The resulting dynamics in which the ports satisfy the algebraic constraint $(f,e)\in D_I$ will be given by a Dirac system on the configuration space $M$. This is best understood by means of an example:

\begin{example}\label{ex:VdS_open1} The following system gives the dynamics of a ``port-controlled generalized Hamiltonian system'' (see~\citep{VdSM2}, page 56):
\begin{align}\label{eq:PortControlled}
\begin{split}
\dot x &= J(x)\frac{\partial E}{\partial x}(x)+g(x)f,\\
e &= g^T(x)\frac{\partial E}{\partial x}(x).
\end{split}
\end{align}
Here the vector $x\in\R^n$ represents the state variables, $J(x)$ is a $n\times n$ skew-symmetric matrix, $g(x)$ is a $n\times m$ linear map, $f\in\R^m$ stand for the flows (inputs) of the system and the function $E$ is the energy of the system. The term $e$ represents the outputs of the system. Finally, $g^T(x)$ denotes the transpose of the matrix $g(x)$. In this example there is an implicit identification between $\R^p$ and its dual $(\R^p)^*$ using the euclidean metric.

To see that~\eqref{eq:PortControlled} is a  \textsc{ph}-system according to Definition~\ref{def:PH-System}, consider the following two vector bundles over $M=\R^n$: $U_1=T\R^n$ and $U_2=\R^n\times \R^m$ (the trivial vector bundle with fiber $\R^m$). Then
\[
D_{(U_1\oplus U_2)}=\{(X, f,\alpha, e)\st J\alpha+g f=X,\; e+g^T\alpha=0\}\subset (U_1\oplus U_2)\oplus (U_1\oplus U_2)^*
\]
defines a Dirac structure in $(U_1\oplus U_2)$ (this is easy to check directly). The corresponding \textsc{ph}-system
\[
(\dot x,f,{\rm d}E,-e)\in D_{(U_1\oplus U_2)}
\]
leads to~\eqref{eq:PortControlled}. In particular, we have that $dE/dt=e^Tf$. Note that, if one thinks of the flow $f$ as an input, the second equation $e= g^T(\partial E/\partial x)(x)$ imposes no constraints on the dynamics, and can be regarded as a definition of the effort in this case.

To interconnect the ports of the system, we choose an interconnecting Dirac structure $D_I\subset U_2\oplus U_2^*$ (this construction will be discussed later in detail). For simplicity, let us take $D_I=U_2\oplus\{0\}$, which leads to the equations:
\begin{align}\label{eq:PortControlled2}
\begin{split}
\dot x &= J(x)\frac{\partial E}{\partial x}(x)+g(x)f,\\
0 &= g^T(x)\frac{\partial E}{\partial x}(x).
\end{split}
\end{align}
We remark that the system~\eqref{eq:PortControlled2} is completely different from~\eqref{eq:PortControlled}: the flows $f$ are now interpreted as multipliers needed in order for the constraint equation $0=g^T(x)\frac{\partial E}{\partial x}(x)$ to be satisfied. Note also that the energy is conserved. Indeed, from the isotropy of $D_I$ it follows that $dE/dt=\langle e,f\rangle=0$. The dynamics given by~\eqref{eq:PortControlled2} can be written as a Dirac system on $TM$ as follows. Consider the Dirac structure
\[
D_{U_1}=\{(X,\alpha)\st \exists f\in\R^m \mbox{ with } J\alpha+g f=X,\; g^T\alpha=0\}\subset U_1\oplus U_1^*.
\]
Then the Dirac system
\[
\dot x\oplus {\rm d}E\in D_{U_1}
\]
is equivalent to~\eqref{eq:PortControlled2}. We will say that we have interconnected or closed the ports of the \textsc{ph}-system~\eqref{eq:PortControlled}.
\end{example}
\vspace{.3cm}

\subsection{Composition of \textsc{PH}-systems}

The operation of closing the ports of a \textsc{ph}-system that we have discussed in Section~\ref{sec:OpenClose} extends to the case where we have multiple \textsc{ph}-systems, resulting in a procedure to interconnect them. This interconnection procedure has been studied under the terminology of \emph{composition of Dirac structures} (see~\citep{CerVdSBan,CerVdSBan2}). The aim of this section is to give a geometric construction of the composition in terms of the backward and forward of Dirac structures. For concreteness, we will focus on the case where we have two \textsc{ph}-systems that we want to interconnect.

We start with the simplest case of vector spaces. Let $U_1,U_2,V_1,V_2$ be finite dimensional vector spaces; we think of $U_2$ and $V_2$ as being the space of flows for \textsc{ph}-systems.  Starting from Dirac structures on $U_1\oplus U_2$ and $V_1\oplus V_2$ (representing \textsc{ph}-systems), and choosing an interconnecting Dirac structure $D_I$ on $U_2\oplus V_2$, we will define a new Dirac structure on $U_1\oplus V_1$ which governs the dynamics of the interconnected system. Note that for the interconnected system there are no ports.

In what follows, we will write $u_i\in U_i$, $v_i\in V_i$, $\alpha_i\in U_i^*$ and $\beta_i\in V_i^*$ to denote elements in $U_i$, $V_i$, and in their duals. We consider the maps
\begin{align*}
\varphi:U_1\times U_2\times V_1\times V_2&\to U_1\times U_2\times U_2\times V_2\times V_1\times V_2,\\
(u_1,u_2,v_1,v_2)&\mapsto (u_1,u_2,u_2,v_2,v_1,v_2),
\end{align*}
and
\begin{align*}
\psi:U_1\times U_2\times V_1\times V_2&\to U_1\times V_1,\\
(u_1,u_2,v_1,v_2)&\mapsto (u_1,v_1).
\end{align*}

\begin{definition}\label{def:composition}\normalfont Let $D_{(U_1\oplus U_2)}$ and $D_{(V_1\oplus V_2)}$ be Dirac structures on $U_1\oplus U_2$ and $V_1\oplus V_2$, respectively, and let $D_I$ be a Dirac structure on $U_2\oplus V_2$. The \emph{composition of $D_{(U_1\oplus U_2)}$ and $D_{(V_1\oplus V_2)}$ via $D_{I}$} is the Dirac structure on $U_1\times V_1$ given by:
\begin{equation}\label{eq:composicion}
D_{(U_1\oplus U_2)}\|_{D_{I}}D_{(V_1\oplus V_2)}=(\F\psi \circ \B\varphi) (D_{(U_1\oplus U_2)}\times D_{I} \times D_{(V_1\oplus V_2)}).
\end{equation}
\end{definition}

Let us first note that $D_{(U_1\oplus U_2)}\times D_{I} \times D_{(V_1\oplus V_2)}$ is a Dirac structure on $U_1\times U_2\times U_2\times V_2\times V_1\times V_2$. Therefore, the composition $D_{(U_1\oplus U_2)}\|_{D_{I}}D_{(V_1\oplus V_2)}$ defines indeed a Dirac structure on $U_1\times V_1$. The coordinate expression of the composition is given in the following proposition, which shows that Definition~\ref{def:composition} is a geometric version of the composition in~\citep{CerVdSBan,CerVdSBan2}:
\begin{proposition}\label{pro:comp} With the notations above, we have
\begin{align*}
D_{(U_1\oplus U_2)}\|_{D_{I}}D_{(V_1\oplus V_2)}=&\{(u_1,v_1,\alpha_1,\beta_1)\st \text{there exist } \; (u_2,v_2,-\alpha_2,-\beta_2)\in D_I \\
& \; \text{such that}\;\;  (u_1,u_2,\alpha_1,\alpha_2)\in D_{(U_1\oplus U_2)}, (v_1,v_2,\beta_1,\beta_2)\in D_{(V_1\oplus V_2)}\}.
\end{align*}
\end{proposition}

\begin{proof}
We write $u_i^*,\bar u_i^*\in U_i^*$ and $v^*_i,\bar v_i^*\in V_i^*$ for covectors. The dual of $\varphi$ is the map
\begin{align*}
\varphi^*:U_1^*\times U_2^*\times U_2^*\times V_2^*\times V_2^*\times V_1^*&\to U_1^*\times U_2^*\times V_1^*\times V_2^*,\\
(u_1^*,u_2^*,\bar u_2^*,\bar v_2^*,v_2^*,v_1^*)&\mapsto (u_1^*,u_2^*+\bar u_2^*,v_1^*,v_2^*+\bar v_2^*).
\end{align*}
It follows that $\B \varphi(D_{(U_1\oplus U_2)}\times D_{I} \times D_{(V_1\oplus V_2)})$ has the following description:
\begin{align*}
\{(u_1,u_2,v_1,v_2,u_1^*,u_2^*+\bar u_2^*,&v_1^*,v_2^*+\bar v_2^*) \;\mbox{ such that } (u_1,u_2,u_1^*,u_2^*)\in D_{(U_1\oplus U_2)}, \\
& (u_2,v_2,\bar u_2^*,\bar v_2^*)\in D_I, \mbox{and }  (v_1,v_2,v_1^*,v_2^*)\in D_{(V_1\oplus V_2)}\}.
\end{align*}
On the other hand, the dual of $\psi$ is given by
\[
\psi^*(u_1^*,v_1^*)=(u_1^*,0,v_1^*,0).
\]
A computation shows that $\left(\F \psi\circ \B \varphi\right) (D_{(U_1\oplus U_2)}\times D_I\times D_{(V_1\oplus V_2)})$ agrees with the expression in~Proposition~\ref{pro:comp}.
\end{proof}

\noindent  As a particular case,  one obtains the following result in~\citep{CerVdSBan} (see also~\citep{VdS-Book}):
\begin{corollary} Let $U_1$, $U_2$ and $U_3$ be finite vector spaces and $D_a$, $D_b$ be Dirac structures in $U_1\times U_2$ and $U_3\times U_2$, respectively. Let $D_I$ be the following Dirac structure on $U_2\times U_2$:
\[
D_I=\{(u_2,\hat u_2,\alpha_2,\hat \alpha_2)\in (U_2\times U_2)\oplus (U_2\times U_2)^*\st \hat u_2=-u_2, \hat \alpha_2=\alpha_2\}.
\]
Then
\begin{align*}
{D_a}\|_{D_{I}}D_b=&\{(u_1,u_3,\alpha_1,\alpha_3)\st  \text{there exists } \; (u_2,\alpha_2)\in U_2\oplus U_2^* \\
 &\text{such that}\;\; (u_1,u_2,\alpha_1,\alpha_2)\in D_a, (u_3,-u_2,\alpha_3,\alpha_2)\in D_b\}.
\end{align*}
\end{corollary}

The composition of Dirac structures in vector spaces can be extended to the case of Dirac structures on vector bundles. In the case of two \textsc{ph}-systems, the composition can be used to interconnect them. Consider two \textsc{ph}-structures $A$ and $B$, where
\[
A=(\pi_{(U_1=TM,M)},\pi_{(U_2,M)},D_{(U_1\oplus U_2)}), \quad B=(\pi_{(V_1=TN,N)},\pi_{(V_2,N)},D_{(V_1\oplus V_2)}).
\]
Given a Dirac structure $D_I$ on the vector bundle $U_2\times V_2\to M\times N$, we define the composition of $A$ and $B$ via $D_I$ fiberwise, using the construction above for the case of vector spaces. That is, at each point, we have
\[
D_{(U_1\times V_1)}(m,n)=D_{(U_1\oplus U_2)}(m)\|_{D_{I}(m,n)}D_{(V_1\oplus V_2)}(n).
\]
The proof that $D_{(U_1\times V_1)}$ defines a Dirac structure on the vector bundle $U_1\times V_1\to M\times N$ can be done mimicking the construction of Definition~\ref{def:composition} in the case of vector bundles. We omit the details here.

Finally, if have two \textsc{ph}-systems with energy functions $E_1\colon M\to \R$ and $E_2\colon N\to \R$, the composition $D_{T(M\times N)}$ leads naturally to a Dirac system on the space $T(M\times N)$
\[
\dot x\oplus {\rm d}(E_1+E_2)\in D_{T(M\times N)},
\]
where $\dot x\in T(M\times N)$, $E_1+E_2\colon M\times N\rightarrow \mathbb{R}$ is defined by $(E_1+E_2)(m,n)=E_1(m)+E_2(n)$.

\begin{example} Take two generalized port-controlled Hamiltonian systems 
\begin{equation*}
\begin{aligned}[c]
\dot x &= J(x)\frac{\partial E}{\partial x}(x)+g(x)f,\\
e &= g^T(x)\frac{\partial E}{\partial x}(x),
\end{aligned}
\qquad\qquad
\begin{aligned}[c]
\dot{\overline{x}} &= \overline{J}(\overline{x})\frac{\partial \overline{E}}{\partial \overline{x}}(\overline{x})+g(\overline{x})\overline{f},\\
\overline{e} &= \overline{g}^T(\overline{x})\frac{\partial \overline{E}}{\partial \overline{x}}(\overline{x}).
\end{aligned}
\end{equation*}
The notation is the same as in Example~\ref{ex:VdS_open1}, in particular we have $M=N=\R^n$, $U_1=\overline{U_1}=T\R^n$ and $U_2=\overline{U_2}=\R^n\times \R^m$. We consider the Dirac structure 
\[
D_I=\{(f,\overline{f},e,\overline{f})\st f=-\overline{f},\;e=\overline{e}\}\subset     (U_2\times U_2)\oplus (U_2\times U_2)^*.
\]
The composition of these two systems via $D_I$ leads to the following Dirac system on $U_1\times V_1\to M\times N$:
\[
\dot x = J(x)\frac{\partial E}{\partial x}(x)+g(x)\lambda,\quad \dot{\overline{x}} = \overline{J}(\overline{x})\frac{\partial \overline{E}}{\partial \overline{x}}(\overline{x})-g(\overline{x})\lambda,\quad g^T(x)\frac{\partial E}{\partial x}(x)=\overline{g}^T(\overline{x})\frac{\partial \overline{E}}{\partial \overline{x}}(\overline{x}).
\]
Note that $\lambda$ is regarded as a multiplier. 
\end{example}
\vspace{.3cm}

\begin{remark} In the literature, the interconnection of multiple (not necessarily two)  Dirac structures by means of an interconnectiong Dirac structure has been explored in~\citep{Generalcomposition}. Similar to the construction above, it is possible to define a composition of $N$ \textsc{ph}-systems in terms of the forward and backward operators which recovers the results in this reference. Since we will describe an alternative geometric approach in the next section, we omit the details. See also Remark~\ref{rem:1} later.
\end{remark}

\section{A geometric framework for Input-Output systems}\label{sec:InputOutput}

In the previous sections we have argued that Dirac structures provide a unified framework to describe a wide class of systems. In particular, we have shown that closed or interconnected systems are typically represented by Dirac systems on the tangent bundle of a manifold, while open systems can be described by the notion of \textsc{ph}-systems. The aim of this section is to introduce a geometric framework which encompasses both closed and open systems, and in terms of which the interconnection of systems will be defined. The construction is based on the observation that open systems can be modeled using a generalization of Dirac systems where the subbundle defining the dynamics is merely coisotropic (rather than both coisotropic and isotropic, as in the Dirac case).

We will explicitly show in this section that the proposed formalism recovers numerous examples in the literature of Dirac structures and port-Hamiltonian systems, but also systems which have not been described in this way before. The forward and backward operators preserve the geometric structures and properties of the initial systems before interconnection, in particular, Dirac structure after composition operations. We plan to use this geometric construction to implement geometric preserving integrators for interconnected Dirac systems.

\subsection{Forward input-output port-Hamiltonian systems}\label{subsec:FPDS}

We start with the notion of \emph{forward input-output (port-Hamiltonian) system}, and the related notion of \emph{open forward input-output (port-Hamiltonian) system}. In the next subsection, we will describe the dual notions of \emph{backward input-output (port-Hamiltonian) system} and  \emph{open backward input-output (port-Hamiltonian) system}.

\begin{definition}\normalfont A \emph{forward input-output structure} (\textsc{fio}-structure) is a 5-uple
\begin{equation}\label{eq:fios_A}
A=\left(\pi_{(U_1,M)},\pi_{(U_2,M)},D_{U_1},D_{U_2},g_{U_2U_1}\right)
\end{equation}
where $\pi_{(U_i,M)}\colon U_i \to M,\,\,i=1,2,$ is a vector bundle, $D_{U_i}$ is a Dirac structure on $U_i,\,\,i=1,2,$ and $g_{U_2U_1}\colon U_2\to U_1$ is a vector bundle map over the identity $1_M:M\to M$.
\end{definition}

We will call \emph{open forward input-output structure} objects obtained by replacing in the definition of a \textsc{fio}-structure $A$ the Dirac structure $D_{U_2}$ by the coisotropic structure $U_2 \oplus U_2^*$. More precisely:

\begin{definition}\normalfont An \emph{open forward input-output structure} (\textsc{ofio}-structure) is a 5-uple
\[
A=\left(\pi_{(U_1,M)},\,\pi_{(U_2,M)},\,D_{U_1},\,U_2 \oplus U_2^*,\,g_{U_2U_1}\right),
\]
where $D_{U_1}$ is a Dirac structure on $U_1$ and $g_{U_2U_1}\colon U_2\to U_1$ is a vector bundle map over the identity.
\end{definition}

Note that the following diagram is commutative:
\begin{figure}[H]
\centering
\includegraphics{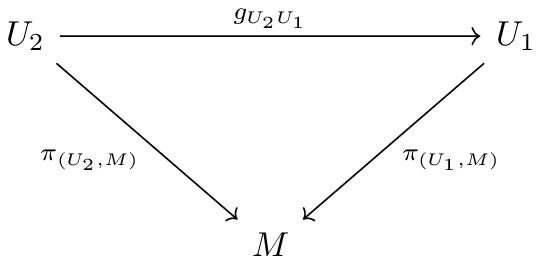}
\end{figure}

Given a \textsc{fio}-structure $A$ as in~\eqref{eq:fios_A}, we associate to it the following Dirac structure on $U_1$:
\begin{equation}\label{eq:D_A}
D_A=\F (\Phi_A)\left(D_{U_1} \oplus D_{U_2}\right),
\end{equation}
where $\Phi_A\colon U_1 \oplus U_2 \to U_1$ is the surjective vector bundle map over $1_M$ given by
\begin{equation}\label{eq:Phi_A}
\Phi_A(u_1 \oplus u_2)=u_1 + g_{U_2U_1}(u_2).
\end{equation}
Explicitly, $D_A$ is the set of all $(u_1,\alpha_1) \in U_1 \oplus U_1^*$ such that there exists $(u_2,\alpha_2) \in U_2 \oplus U_2^*$ with
\begin{equation}\label{eq:fio_explicit}
\left.
\begin{aligned}
(u_1-g_{U_2U_1}(u_2),\alpha_1) & \in D_{U_1}, \\
(u_2,\alpha_2) & \in D_{U_2}, \\
g^*_{U_2U_1}(\alpha_1)& = \alpha_2.
\end{aligned}
\; \right\}
\end{equation}

In the same way, given an \textsc{ofio}-structure $A$, we associate to it the following coisotropic structure on $U_1$:
\begin{equation}\label{eq:Sigma_A}
\Sigma_A=\F (\Phi_A)\left(D_{U_1} \oplus (U_2 \oplus U^*_2)\right)
\end{equation}
where the map $\Phi_A$ is given by~\eqref{eq:Phi_A}. The fact that $\Sigma_A$ defines a coisotropic structure follows directly from the fact that $D_{U_1}\oplus (U_2 \oplus U_2^*)$ is a coisotropic structure on $U_1\oplus U_2$ (see Lemma~\ref{lem:inclussions}). We deduce that $\Sigma_A$ is the set of all $(u_1,\alpha_1) \in U_1 \oplus U_1^*$ such
that there exists $(u_2,\alpha_2) \in U_2 \oplus U_2^*$ with
\begin{equation}\label{eq:ofio_explicit}
\left.
\begin{aligned}
(u_1-g_{U_2U_1}(u_2),\alpha_1) & \in D_{U_1}, \\
g^*_{U_2U_1}(\alpha_1)& = \alpha_2,
\end{aligned}
\; \right\}
\end{equation}
or equivalently, since $(u_2, \alpha_2) \in U_2 \oplus U_2^*$ is arbitrary,
\begin{equation*}
(u_1-g_{U_2U_1}(u_2),\alpha_1) \in D_{U_1}.
\end{equation*}

\begin{example}\label{ex:Bivectorforward}
Consider the particular case in which $D_{U_1}$ is the Dirac structure associated to a bivector
$\Lambda$ on $U_1$ and $D_{U_2} = U_2 \oplus \{0\}$. Then the system~\eqref{eq:fio_explicit} reads
\begin{align*}
u_1-g_{U_2U_1}(u_2) &= \sharp_\Lambda (\alpha_1),\\
g^*_{U_2U_1}(\alpha_1) & = 0,
\end{align*}
or, equivalently,
\begin{align*}
u_1 - \sharp_\Lambda (\alpha_1) & \in \operatorname{Im}g_{U_2U_1}, \\
\alpha_1 &\in (\operatorname{Im}g_{U_2U_1})^\circ.
\end{align*}
From this we can deduce that $F^{(*)}_{D_A} = (\operatorname{Im}\,g_{U_2U_1})^\circ\subset U_1^*$ and that $D_A$ is the Dirac structure on $U_1$ determined by the codistribution
$(\operatorname{Im}\,g_{U_2U_1})^\circ$ and the restriction of $\Lambda$ to $(\operatorname{Im}\,g_{U_2U_1})^\circ$. Note that from the equations above it follows that $F_{D_A} = \sharp_\Lambda (\operatorname{Im}\,g_{U_2U_1})^\circ + \operatorname{Im}\,g_{U_2U_1}$.

A partial converse for this example can be easily proven under regularity conditions. Given any Dirac structure $D$ on $U_1$
such that $F^{(*)}_D$ is a subbundle of $U_1^*$, there exists a bivector $\Lambda$ on $U_1$, a vector bundle $U_2$ and a vector bundle map $g_{U_2U_1}\colon  U_2 \to U_1$ such that for the \textsc{fio}-structure
\[
A=(\pi_{(U_1,M)},\pi_{(U_2,M)},D_{U_1}=D_\Lambda,D_{U_2}=U_2\oplus \{0\},g_{U_2U_1})
\]
we have $D_A=D$. In order to achieve this, first we must take $U_2$ and $g_{U_2U_1}$ in such a way that the condition
$F^{(*)}_D = (\operatorname{Im}\,g_{U_2U_1})^\circ$ is satisfied. By assumption, $F^{(*)}_D$ is a subbundle, so there exists a bivector $\tilde\Lambda : F^{(*)}_D \times F^{(*)}_D \to \mathbb{R}$ on $F_D$ defining the Dirac structure $D$, which we can extend to a bivector $\Lambda \colon  U_1^* \times U_1^* \to \mathbb{R}$. The fact that $D_A=D$ is easy to check.
\end{example}
\vspace{.3cm}

The most important case of an \textsc{fio}-structure (or \textsc{ofio}-structure) $A$ occurs when $U_1 = TM$. For a given energy function $E\colon M\to \mathbb{R}$ and using~\eqref{eq:D_A} and~\eqref{eq:Sigma_A}, we will define the dynamics by the following equations:
\begin{align*}
(x,\dot x) \oplus {\rm d}E(x) \in D_A, \qquad &\mbox{if $A$ is a \textsc{fio}-structure}.\\
(x,\dot x) \oplus {\rm d}E(x) \in \Sigma_A, \qquad &\mbox{if $A$ is an \textsc{ofio}-structure}.
\end{align*}
In the case of \textsc{fio}-structure~\eqref{eq:fio_explicit}, these equations become
\begin{align*}
((x,\dot x)-g_{U_2U_1}(u_2),{\rm d}E(x)) & \in D_{U_1}, \\
(u_2,\alpha_2) & \in D_{U_2}, \\
g^*_{U_2U_1}({\rm d}E(x))& = \alpha_2.
\end{align*}
Similarly, for the case of an \textsc{ofio}-structure~\eqref{eq:ofio_explicit} the equations of motion are:
\begin{align*}
((x,\dot x)-g_{U_2U_1}(u_2),{\rm d}E(x)) & \in D_{U_1},\\
g^*_{U_2U_1}{\rm d}E(x)& = \alpha_2.
\end{align*}
Note that in the case of an \textsc{ofio}-structure, the second equation $g^*_{U_2U_1}{\rm d}E(x) = \alpha_2$ does not impose any restriction on the dynamics of the state variable $x$. It does, however, have some physical meaning in concrete examples as we will give an idea in the examples below.

If $A$ is a \textsc{fio}-structure and $E$ is an energy function, the system $(x,\dot x) \oplus {\rm d}E(x) \in D_A$ is the Dirac system naturally associated to the \textsc{fio}-structure $A$ and the energy function $E$. It will be called a \emph{forward input-output system} (\textsc{fio}-system). It is important to observe that the system $(x,\dot x) \oplus {\rm d}E(x) \in \Sigma_A$ (which occurs in case $A$ is an \textsc{ofio}-structure) is not a Dirac system but a \emph{coisotropic system} since $\Sigma_A$ is a coisotropic structure.

\begin{example} In the case $g_{U_2U_1}=0$, the equations of an~\textsc{fio}-system become
\[
((x,\dot x),{\rm d}E(x)) \in D_{U_1}.
\]
Therefore the theory of Dirac systems is contained in the theory of \textsc{fio}-systems.
\end{example}
\vspace{.3cm}

\begin{example}\label{ex:VdS_open} Let $A$ be the \textsc{ofio}-structure with $D_{U_1}$ given by the graph of a Poisson bivector $\Lambda$ on $M$. Then the dynamics is governed by the equation
\[
(x, \dot{x})= \sharp_\Lambda ({\rm d}E (x)) + g_{U_2U_1}(u_2)
\]
which is simply a system with control parameters $u_2$. The interpretation is that of a system with ``open ports'' which might be used to model the various interactions of the system. The terminology ``\textsc{ofio}-system'' is motivated by this observation.

The equations of motion of an \textsc{ofio}-system should be compared with those of a ``port-controlled generalized Hamiltonian system'' that we have already discussed in the realm of \textsc{ph}-systems (Example~\ref{eq:PortControlled}):
\begin{align*}
\dot x &= J(x)\frac{\partial E}{\partial x}(x)+g(x)f,\\
e &= g^T(x)\frac{\partial E}{\partial x}(x).
\end{align*}
Here $J$ plays the role of $\sharp_\Lambda$, $f$ stands for the flows (inputs) of the system and $E=H$ is the Hamiltonian. The term $e$ represents the efforts (outputs) of the system as it evolves according to the equation $\dot x = J(\partial E/\partial x)+gf$. This is one possible physical meaning for the ``non-dynamical equation'' $g^*_{U_2U_1}{\rm d}E(x) = \alpha_2$ discussed above.
\end{example}
\vspace{.3cm}

\begin{example}\label{ex:VdS_closed} Let $A$ be the \textsc{fio}-structure with $D_{U_1}$ given by the graph of a Poisson bivector $\Lambda$ on $M$ and $D_{U_2}=U_2\oplus \{0\}$. Then the dynamics is described by
\begin{align*}
(x, \dot{x})
&= \sharp_\Lambda ({\rm d}E (x)) + g_{U_2U_1}(u_2),\\
g^*_{U_2U_1}({\rm d}E(x))& = 0,
\end{align*}
which is a differential-algebraic equation (DAE).

As a particular case, we obtain, in the language of~\citep{DVdS}, the so-called ``representation II'' of the generalized Hamiltonian system with Hamiltonian $H=E$. In~\citep{DVdS}, the equations read (see page 64)
\begin{align*}
\dot x &= J(x)\frac{\partial E}{\partial x}(x)+g(x)f,\\
0 &= g^T(x)\frac{\partial E}{\partial x}(x).
\end{align*}
See the paragraph after~\eqref{eq:PortControlled2} in Example~\ref{ex:VdS_open1} for more details.
\end{example}
\vspace{.3cm}

\begin{example} Consider a spring pendulum on the plane which is attached to a fixed peg $O$ as depicted in Diagram~\ref{dia:pendulo}. Assume that a fixed force $F\in \R^2$ acts on the mass $m$. We will show how to describe this system as an \textsc{ofio}-system.

\begin{figure}[H]
\centering
\includegraphics{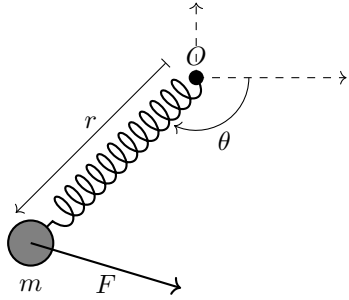}
\caption{A force $F$ acting on a spring pendulum}\label{dia:pendulo}
\end{figure}

We take polar coordinates $q=(r,\theta)$ as shown in Diagram~\ref{dia:pendulo}, where the angle $\theta$ is measured with respect to a fixed frame at $O$. The configuration space is $T(\R\times S^1)=TQ$. The Lagrangian for the spring pendulum is the function on $T(\R\times S^1)$ given by:
\[
L(q,v)=\frac{1}{2}m(v_r^2+r^2v_\theta^2)-\frac{1}{2}k(r-r_0)^2,
\]
where $k$ is the constant of the spring and $r_0$ denotes its natural length. We have the following energy function on $M=TQ\oplus T^*Q$:
\[
E(q,v,p)=pv-L(q,v)=p_\theta v_\theta + p_r v_r - \frac{1}{2}m(v_r^2+r^2v_\theta^2)+\frac{1}{2}k(r-r_0)^2,
\]
We denote by $D_1$ the usual Dirac structure on $TM$ obtained by pulling back $\omega_Q$ to $M$. We let $g_A:M\times\R^2\to TM$ be the vector bundle map:
\[
g_A(q,v,p,F)=(q,v,p,0,0,F),
\]
where $F\in\R^2$. If we define the following \textsc{ofio}-structure
\[
A=\left(TM,M\times \R^2,D_1,(M\times \R^2)\oplus (M\times \R^2)^*,g_A\right),
\]
the equations of motion corresponding to $A$ and the energy $E$ are (recall that $x=(q,v,p)$):
\begin{align*}
(\dot x-g_A(F),{\rm d}E(x)) & \in D_1,\\
g_A^*({\rm d}E(x))& = \alpha,
\end{align*}
where $(F,\alpha)\in \R^2\times (\R^2)^*$. If we write $F=(F_r,F_\theta)$, the dynamical equation $(\dot x-g_A(F),{\rm d}E(x)) \in D_1$ represents the dynamics of the spring pendulum acted by a force with components $F_r$ and $F_\theta$ along the radial ($r$) and angular ($\theta$) directions. For instance, it is immediate to check that the evolution equations for the momenta $p_r$ and $p_\theta$ are:
\[
\dot p_r - F_r= \frac{\partial L}{\partial r},\qquad \dot p_\theta - F_\theta= \frac{\partial L}{\partial \theta}.
\]
\end{example}
\vspace{.3cm}

\begin{example}[\textbf{Nonholonomic mechanics}]\label{ex:nonh-ex} Let $\D\subset TQ$ be a vector subbundle and consider the vector bundle $U_2=T^*Q\oplus \D^\circ$ over $T^*Q$. We define the vector bundle morphism $g_A\colon U_2\to TT^*Q$ over the identity in $T^*Q$ by:
\[
g_A(\alpha_q, \mu_q)=\frac{d}{dt}\Big|_{t=0}(\alpha_q+t \mu_q)\in T_{\alpha_q}T^*Q,
\]
where $\alpha_q\in T^*Q$, $\mu_q\in \D^\circ$. Observe that by construction $g(\alpha_q, \mu_q)$ is a vertical element, i.e. $g_A(\alpha_q, \mu_q)\in V_{\alpha_q}T^*Q$. Given a basis of sections $\{\mu^{a}\}$ of $\D^\circ\rightarrow Q$, we have the following coordinate representation of the mapping $g_A$:
\[
g_A(q^i, p_i,\lambda_a)=(q^i, p_i, 0, \lambda_a \mu^a_i(q))\equiv \lambda_a \mu^a_i(q)\frac{\partial}{\partial p_i}.
\]
We consider the \textsc{fio}-structure
\[
A=(TT^*Q,U_2,D_1,D_{U_2},g_A),
\]
with $D_1\subset TT^*Q\oplus T^*T^*Q$ the usual Dirac structure induced by the graph of the canonical symplectic form $\omega_Q$ and $D_{U_2}=U_2\oplus 0$. For a Hamiltonian $H\colon T^*Q\to\R$ the associated \textsc{fio}-system gives the following system of equations:
\begin{align*}
\frac{d{q}^i}{dt}&=\frac{\partial H}{\partial  p_i},\\
 \frac{d{p}_i}{dt}&=-\frac{\partial H}{\partial  q_i}+\lambda_a \mu^a_i(q), \\
 0&=\mu^a_i(q)\frac{\partial H}{\partial  p_i}.
\end{align*}
These are precisely the equations of a nonholonomic system defined by a Hamiltonian   $H\colon T^*Q\to\R$ and a nonholonomic distribution $\D$. See e.g.~\citep{Marle} for more details. 

This shows that the interconection of different Dirac structures may lead to nonholonomic constraints and their corresponding reaction forces. 
\end{example}

\vspace{.3cm}

%
%

\subsection{Backward input-output port-Hamiltonian systems}\label{subsec:BPDS}
We now turn to the notions of \emph{backward input-output (port-Hamiltonian) systems} and \emph{open-backward (port-Hamiltonian) systems}. The definitions and results are very similar to those of the forward case. Statements related to ``backward'' are dual from those related to ``forward'' and could be obtained one from each other directly, but we will describe them separately in detail because the concrete expressions that appear in each case are useful in particular examples and help to relate the results to those in the literature. See Appendix~\ref{ap:B} for more details.

\begin{definition}\normalfont A \emph{backward input-output structure} (\textsc{bio}-structure)  is a 5-uple
\begin{equation*}\label{eq:cbpds_A}
A=\left(\pi_{(U_1,M)},\pi_{(U_2,M)},D_{U_1},D_{U_2},p_{U_1U_2}\right)
\end{equation*}
where $\pi_{(U_i,M)}\colon U_i \to M,\,\,i=1,2,$ is a vector bundle, $D_{U_i}$ is a Dirac structure on $U_i,\,\,i=1,2,$ and $p_{U_1U_2}\colon U_1\to U_2$ is a vector bundle map over the identity $1_M$.
\end{definition}

\begin{definition}\normalfont An \emph{open backward input-output structure} (\textsc{obio}-structure) is a 5-uple
\[
A=\left(\pi_{(U_1,M)},\,\pi_{(U_2,M)},\,D_{U_1},\,U_2 \oplus U_2^*,\,p_{U_1U_2}\right).
\]
where $D_{U_1}$ is a Dirac structure on $U_1$ and $p_{U_1U_2}\colon U_1\to U_2$ is a vector bundle map over the identity.
\end{definition}

We will simply write \textsc{bio}-structure or \textsc{obio}-structure. The following diagram is commutative:
\begin{figure}[H]
\centering
\includegraphics{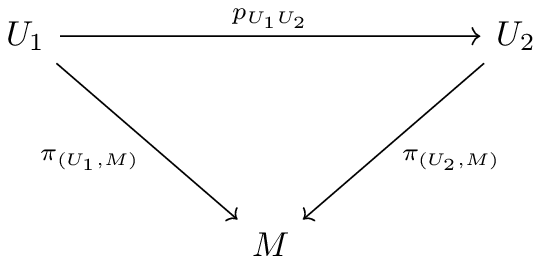}
\end{figure}

Given a \textsc{bio}-structure $A$, we associate to it the following Dirac structure on $U_1$:
\begin{equation}\label{eq:D_A_back}
D_A=\B (\Psi_A)\left(D_{U_1} \oplus D_{U_2}\right),
\end{equation}
where $\Psi_A\colon U_1 \to U_1\oplus U_2 $ is the injective vector bundle map over $1_M$ given by $\Psi_A(u_1)=u_1\oplus p_{U_1U_2}(u_1)$. Note that $D_A$ is the set of all $(u_1,\alpha_1) \in U_1 \oplus U_1^*$ such that there exists $(u_2,\alpha_2) \in U_2 \oplus U_2^*$ with
\begin{equation}\label{eq:bio_explicit}
\left.
\begin{aligned}
(u_1,\alpha_1-p_{U_1U_2}^*(\alpha_2)) & \in D_{U_1}, \\
(u_2,\alpha_2) & \in D_{U_2}, \\
p_{U_1U_2}(u_1)& = u_2.
\end{aligned}
\; \right\}
\end{equation}

Likewise, associated to an \textsc{obio}-structure $A$ we consider the coisotropic structure on $U_2$ given by:
\begin{equation}\label{eq:Sigma_A_back}
\Sigma_A=\B (\Psi_A)\left(D_{U_1} \oplus (U_2 \oplus U^*_2)\right)
\end{equation}
where  the map $\Psi_A\colon  U_1\to U_1\oplus U_2$ is defined as above. One checks that $\Sigma_A$ is the set of all $(u_1,\alpha_1) \in U_1 \oplus U_1^*$ such
that there exists $(u_2,\alpha_2) \in U_2 \oplus U_2^*$ with
\begin{equation}\label{eq:obio_explicit}
\left.
\begin{aligned}
(u_1,\alpha_1-p_{U_1U_2}^*(\alpha_2)) & \in D_{U_1}, \\
p_{U_1U_2}(u_1)& = u_2,
\end{aligned}
\; \right\}
\end{equation}
or, taking into account that $(u_2, \alpha_2) \in U_2 \oplus U_2^*$ is arbitrary,
\begin{equation*}
(u_1,\alpha_1-p_{U_1U_2}^*(\alpha_2)) \in D_{U_1}.
\end{equation*}

\begin{example} This example is the dual of Example~\ref{ex:Bivectorforward}. Let $D_{U_1}$ be the Dirac structure associated to a 2-form $\omega$ on $U_1$ and $D_{U_2} = \{0\} \oplus U^*_2$. The equations~\eqref{eq:bio_explicit} are
\begin{align*}
\omega^\flat (u_1) &= \alpha_1-p_{U_1U_2}^*(\alpha_2),\\
p_{U_1 U_2}(u_1)& = 0,
\end{align*}
or, using that $(\operatorname{Im}\,p_{U_1U_2}^*)^\circ = \ker\,p_{U_1U_2}$,
\begin{align*}
\omega^\flat (u_1) - \alpha_1 & \in (\ker\,p_{U_1U_2})^\circ, \\
u_1 &\in \ker\,p_{U_1U_2}.
\end{align*}

This implies that $F_{D_A} = \ker\,p_{U_1U_2}$. Then  $D_A$ is the Dirac structure on $U_1$ determined by the distribution
$\ker\,p_{U_1U_2}$ and the restriction of $\omega$ to a 2-form on
it. We also see that $F^{(*)}_{D_A} = \omega^\flat (\ker\,p_{U_1U_2}) + (\ker\,p_{U_1U_2})^\circ$. It is possible to prove a partial converse in the same way as in Example~\ref{ex:Bivectorforward}.
\end{example}
\vspace{.3cm}

\begin{example}[\textbf{Tensor product}] Consider a \textsc{bio}-structure $A$ with $U_1=U_2=TM$, and we denote the Dirac structures by $D_1$ and $D_2$, respectively. Take the map $p_{U_1U_2}:TM\to TM$ to be the identity. Then~\eqref{eq:bio_explicit} becomes
\begin{equation*}
\left.
\begin{aligned}
(u_1,\alpha_1-\alpha_2) & \in D_{1}, \\
(u_1,\alpha_2) & \in D_{2},
\end{aligned}
\; \right\}
\end{equation*}
This coincides with the so-called tensor product $D_1\boxtimes D_2$ of the Dirac structures $D_1$ and $D_2$ (see~\citep{Gualtieri}), which can be alternatively obtained as follows. Consider the diagonal embedding $d:M\to M\times M$, then:
\[
D_1\boxtimes D_2=\B(Td)(D_1\times D_2)\subset TM\oplus T^*M.
\]
This construction is also known by some authors as the bowtie product of $D_1$ and $D_2$ (denoted $D_1\bowtie D_2$). For applications of the tensor product in the interconnection of Dirac structures, we refer to~\citep{JacYo} and references therein.
\end{example}
\vspace{.3cm}

Let $A$ be an \textsc{bio}-structure or \textsc{obio}-structure with $U_1 = TM$. If $E\colon M\to \mathbb{R}$ is an energy function, and using~\eqref{eq:D_A_back} and~\eqref{eq:Sigma_A_back} the dynamics is given by:
\begin{align*}
(x,\dot x) \oplus {\rm d}E(x) \in D_A, \qquad &\mbox{if $A$ is a \textsc{bio}-structure}.\\
(x,\dot x) \oplus {\rm d}E(x) \in \Sigma_A, \qquad &\mbox{if $A$ is an \textsc{obio}-structure}.
\end{align*}
The equations of motion in the case of a \textsc{bio}-structure~\eqref{eq:bio_explicit} are
\begin{align*}
((x, \dot{x}),{\rm d}E(x)-p_{U_1U_2}^*(\alpha_2)) & \in D_{U_1}, \\
(u_2, \alpha_2) & \in D_{U_2},\\
p_{U_1 U_2}(x, \dot{x})& = u_2.
\end{align*}
For an \textsc{obio}-structure~\eqref{eq:obio_explicit} one obtains:
\begin{align*}
((x, \dot{x}),{\rm d}E(x)-p_{U_1U_2}^* (\alpha_2))& \in D_{U_1}, \\
p_{U_1 U_2}(x, \dot{x})& = u_2.
\end{align*}
We observe that in this case the equation $p_{U_1 U_2}(x, \dot{x}) = u_2$ does not impose any restriction on the dynamics.

The system $(x,\dot x) \oplus {\rm d}E(x) \in D_A$ is the Dirac system associated to the \textsc{bio}-structure $A$ and the energy function $E$. It will be called a \emph{backward input-output system} (\textsc{bio}-system). The system $(x,\dot x) \oplus {\rm d}E(x) \in \Sigma_A$ (when $A$ is an \textsc{ofio}-structure) is a coisotropic system.

\begin{example} Let $A$ be the \textsc{bio}-structure with $U_1=TM$ and $D_{U_1}$ given by the graph of a presymplectic form $\omega$ on $M$, and $D_{U_2}=\{0\}\oplus U_2^*$. The equations are:
\begin{align*}
\omega^\flat(\dot{x}) &= {\rm d}E(x)-p_{U_1U_2}^*(\alpha_2),\\
p_{U_1U_2}(x, \dot{x})& = 0,
\end{align*}
which is a DAE. This is the so-called ``representation III'' on~\citep{DVdS}.
\end{example}
\vspace{.3cm}

\subsection{Interconnection of input-output port-Hamiltonian systems}\label{subsec:Inter} In the previous sections we have shown that an \textsc{ofio}/\textsc{obio}-system serves as a model for a dynamical system with open ports, i.e. a system for which interaction with other systems is possible.  We will now describe how, given an interconnecting Dirac structure, it is possible to connect a family of \textsc{ofio}-systems (\textsc{obio}-systems) through the ports in such a way that the resulting system is a \textsc{bio}-system (\textsc{fio}-system) which represents the dynamics of the interconnected system. We will do this via some illustrative examples that have been considered in the literature with different methods~\citep{VdSM2,DVdS}.

\paragraph{Interconnection of \textsc{ofio}-systems.} We start with the forward case. Let
\begin{equation*}
A_i=(\pi_{(U_{1,i},M_i)},\,\pi_{(U_{2,i},M_i)},\,D_{U_{1,i}},\,U_{2,i} \oplus U_{2,i}^*,\,g_{U_{2,i}U_{1,i}}),\quad i = 1,\dots,N,
\end{equation*}
be a family of \textsc{ofio}-structures. If we consider the product manifold $M=M_1\times M_2\times \dots\times M_N$, we can define vector bundles $U_1\to M$ and $U_2\to M$ by
\[
U_1=U_{1,1}\times U_{1,2}\times...\times U_{1,N},\qquad U_2=U_{2,1}\times U_{2,2}\times...\times U_{2,N}.
\]
Using that $D_{U_{1,i}}\subset U_{1,i}\oplus U^*_{1,i}$ is a Dirac structure for each $i=1,\dots,N$, it is easy to verify that
\[
D_{U_1} = D_{U_{1,1}}\times D_{U_{1,2}}\times \dots\times D_{U_{1,N}}\subset U_1\oplus U_1^*
\]
defines a Dirac structure on $U_1$. Finally we can also construct a vector bundle map $g_{U_2U_1}\colon  U_2\to U_1$ as follows:
\[
g_{U_2U_1} = g_{U_{2,1}U_{1,1}}\times g_{U_{2,2}U_{1,2}}\times \dots\times g_{U_{2,N}U_{1,N}}.
\]
With these notations, we define the \emph{product \textsc{ofio}-structure of the family $\{A_i\}_{i}$}, denoted $A_1 \times A_2\times\dots\times A_N$, to be the \textsc{ofio}-structure given by:
\begin{equation*}
A_1 \times A_2\times\dots\times A_N=(\pi_{(U_1,M)},\pi_{(U_2,M)},D_{U_1},U_2 \oplus U_2^*,g_{U_2U_1}).
\end{equation*}
\begin{definition}\label{def:ofio-interconnection}\normalfont With the notations above: Given a Dirac structure $D_{U_2}$ on $U_2$, the \emph{interconnection of the \textsc{ofio}-structures $A_1,\dots A_N$ by $D_{U_2}$} is the \textsc{fio}-structure
\begin{equation}\label{eq:interconnection_ofio}
(\pi_{(U_1,M)},\pi_{(U_2,M)},D_{U_1},D_{U_2},g_{U_2U_1}).
\end{equation}
\end{definition}
The case of greatest interest occurs when each $A_i$ has associated a dynamical system to be interconnected. In this case, we have $U_{1,i}=TM_i$, and there are energy functions $E_i\colon M_i\to \mathbb{R}$. Then we identify $U_1=TM$, and we can consider the dynamics given by the total energy $E=E_1+\dots+E_N$ on $M$ (here it is understood that each energy $E_i$ is pulled back to $M$ via the projection $M\to M_i$).

\begin{example} Take $N=1$, that is we want to interconnect a single \textsc{ofio}-system. Consider the system in Example~\ref{ex:VdS_open}:
\begin{align*}
\dot x &= J(x)\frac{\partial E}{\partial x}(x)+g(x)f,\\
e &= g^T(x)\frac{\partial E}{\partial x}(x).
\end{align*}
In our language, the flows $f$ and efforts $e$ are such that  $(f,e)\in U_2\oplus U^*_2$. To interconnect the system, we choose a Dirac structure $D\subset U_2\oplus U^*_2$. One choice is to set the efforts to zero, namely to consider $D=U_2\oplus \{0\}$, and then the resulting dynamics is precisely that of Example~\ref{ex:VdS_closed}. We will say that the ports have been ``interconnected'' or ``closed''. We point out again that Dirac structures $D\subset U_2\oplus U^*_2$ represent \emph{power-conserving} interconnections (see~\citep{VdSM2,DVdS}).
\end{example}
\vspace{.3cm}
\begin{example} Consider now a family of $N$ systems as in~Example~\ref{ex:VdS_open}, each of them with flows and efforts $(f_i,e_i)\in U_{2,i}\oplus U^*_{2,i}$, $i=1,\dots,N$. To interconnect them one uses a chosen Dirac structure on $U_2=U_{2,1}\times\dots\times U_{2,N}$ as in~\eqref{eq:interconnection_ofio}. This is the geometric version of Proposition~2.2 in~\citep{DVdS} (page 59) within the framework of \textsc{ofio}-systems.

\end{example}
\vspace{.3cm}

\paragraph{Interconnection of an \textsc{obio}-systems.} The interconnection of \textsc{obio}-systems is completely analogous (in fact, is essentially equivalent, see Appendix~\ref{ap:B}) to the forward case. Let
\begin{equation*}
A_i=(\pi_{(U_{1,i},M_i)},\,\pi_{(U_{2,i},M_i)},\,D_{U_{1,i}},\,U_{2,i} \oplus U_{2,i}^*,\,p_{U_{1,i}U_{2,i}}),\quad i = 1,\dots,N,
\end{equation*}
be a family of \textsc{obio}-structures. With the same notations as in the forward case, we consider the manifold $M$, the vector bundles $U_1\to M$ and $U_2\to M$, and the Dirac structure $D_{U_1}$. The map $p_{U_1U_2}\colon U_1\to U_2$ is defined by
\[
p_{U_1U_2} = p_{U_{1,1}U_{2,1}}\times p_{U_{1,2}U_{2,2}}\times \dots\times p_{U_{1,N}U_{2,N}}.
\]
With these notations, we define the \emph{product structure of the family $\{A_i\}_{i}$}, denoted $A_1 \times A_2\times\dots\times A_N$, to be the \textsc{obio}-structure given by:
\begin{equation*}
A_1 \times A_2\times\dots\times A_N=(\pi_{(U_1,M)},\pi_{(U_2,M)},D_{U_1},U_2 \oplus U_2^*,p_{U_1U_2}).
\end{equation*}
\begin{definition}\normalfont With the notations above: Given a Dirac structure $D_{U_2}$ on $U_2$, the \emph{interconnection of the \textsc{obio}-structure $A_1,\dots A_N$ by $D_{U_2}$} is the \textsc{bio}-structure
\begin{equation*}
(\pi_{(U_1,M)},\,\pi_{(U_2,M)},\,D_{U_1},\,D_{U_2},\,p_{U_1U_2}).
\end{equation*}
\end{definition}
When each $A_i$ has dynamics given by the energy function $E_i$, then on $U_1=TM$  we consider the dynamics given by the total energy $E=E_1+\dots+E_N$ on $M$.

\begin{example}\label{ex:LCcir0} Consider the following LC circuit, with two inductors ($L_1$ and $L_3$) and two capacitors ($C_2$ and $C_4$):
\begin{figure}[H]
\centering
\includegraphics{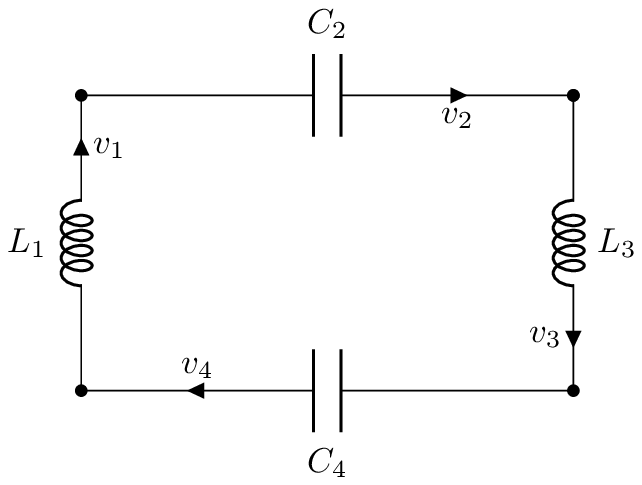}
\caption{The circuit of Example~\ref{ex:LCcir0}}\label{dia:cir0}
\end{figure}

To obtain the Dirac system associated to the circuit we will follow the formalism on~\citep{CEF_DiracConstraints} which provides equations in the tangent bundle of $M=TQ\oplus T^*Q$.

We label the branches according to the numbering of inductors and capacitors: for instance the branch ``$3$'' is the branch with the inductor $L_3$. The configuration space is a vector space $Q$ (the charge space), and an element $q=(q_1,q_2,q_3,q_4)\in Q$ represents charges in the branches $1$, $2$, $3$ and $4$ respectively. The Lagrangian on $TQ\simeq Q\times Q$ is given by:
\[
L(q,v)=\frac{1}{2}L_1v_1^2+\frac{1}{2}L_3v_3^2-\frac{1}{2}\frac{q_2^2}{C_2}-\frac{1}{2}\frac{q_4^2}{C_4},
\]
where $v_i$ ($i=1,2,3,4$) represents the currents in each branch. The sign convention for the currents in each branch is described in Diagram~\ref{dia:cir0}. The energy on $M=TQ\oplus T^*Q\simeq Q\times Q\times Q$ reads:
\[
E(q,v,p)=pv-L(q,v).
\]
The KCL (Kirchhoff's Current Law) gives raise to the distribution $\Delta\subset TQ\simeq Q\times Q$ which is independent of $q$. We have $\Delta=Q\times \Delta_q$ with
\[
\Delta_q=\{v\in T_qQ\st v_1-v_4=0, v_3-v_2=0, v_4-v_3=0\}.
\]
Its annihilator $\Delta_q^\circ\subset T^*_qQ$, representing KCL (Kirchhoff's Voltage Law), is then
\[
\Delta_q^\circ=\{p\in T_q^*Q\st p_1=p_2=p_3=p_4\}.
\]
Consider the distribution in $M$ given by $\Delta_M=T\bar \tau^{-1}(\Delta)$, where $\bar \tau:M\to Q$ is the projection. It is easy to see that $\Delta_M=\{(q,v,p,\dot q,\dot v,\dot p)\st \dot q\in\Delta\}$, i.e. we have the invariant distribution (independent of $x\in M$)
\[
\Delta_M=\{(q,v,p,\dot q,\dot v,\dot p)\st \dot q_1-\dot q_4=0, \dot q_3-\dot q_2=0, \dot q_4-\dot q_3=0 \}.
\]
Let $D_1\subset TM\oplus T^*M$ be the Dirac structure determined by the pullback of $\omega_Q$ to $M$ acting on the constraint distribution $\Delta_M$ (see Theorem~\ref{Thm:DiracStructM}). It is not hard to check that $D_1$ has the following description:
\[
D_1=\left\{(q,v,p,\dot q,\dot v,\dot p,\alpha,\gamma,\beta)\st \dot q\in\Delta, \dot p +\alpha\in \Delta^\circ,\, \gamma=0,\, \dot q-\beta=0\right\}.
\]
The equations of the circuit are then given by the Dirac system on $M$ with Dirac structure $D_1$ and energy function $E$,
\[
(x,\dot x)\oplus {\rm d}E\in D_1,
\]
where $x=(q,v,p)$. It follows that the equations of motion are:
\[
\dot q\in\Delta,\qquad \dot q=v,\qquad p=\frac{\partial L}{\partial v},\qquad \dot p-\frac{\partial L}{\partial q}\in\Delta^\circ.
\]

Assume now that we want to attach 2 ports to the circuit as indicated in Diagram~\ref{dia:cir1} (left). To model this open system, we define the map $p_A:TM\to M\times\R^2$ given by:
\[
p_A(q,v,p,\dot q,\dot v,\dot p)=(q,v,p,\dot q_3-\dot q_2,\dot q_1-\dot q_4).
\]
We then have the following \textsc{obio}-structure
\[
A=\left(TM,M\times \R^2,D_1,(M\times \R^2)\oplus (M\times \R^2)^*,p_A\right).
\]
Note that the dual map $(p_A)^*:(M\times\R^2)^*\simeq M\times\R^2\to T^*M$ is given by
\[
(p_A)^*(q,v,p,e_1,e_2)=(q,v,p,e_1(\alpha_3-\alpha_2)+e_2(\alpha_1-\alpha_4),0,0),
\]
(recall that $\alpha_i$ is the dual of $\dot q^i$).

\begin{figure}[H]
\centering
\includegraphics{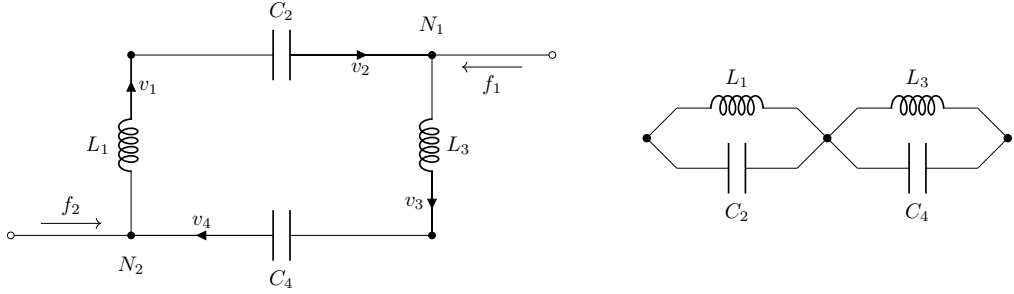}
\caption{The circuit before (left) and after (right) closing the ports}\label{dia:cir1}
\end{figure}

If we denote an element in $(M\times \R^2)\oplus (M\times \R^2)^*$ by $(x,f_1,f_2,e_1,e_2)$ (with $x=(q,v,p)\in M)$, the dynamics of $A$ have the form (see Section~\ref{subsec:BPDS})
\begin{align*}
((x, \dot{x}),{\rm d}E(x)-p_{A}^*(e_1,e_2)) & \in D_1, \\
p_{A}(x, \dot{x})& = (f_1,f_2),
\end{align*}
i.e.
\begin{align*}
\big((x, \dot{x}),{\rm d}E(x)-e_1(\alpha_3-\alpha_2)-e_2(\alpha_1-\alpha_4)\big) & \in D_1, \\
(\dot q_3-\dot q_2,\dot q_1-\dot q_4)& = (f_1,f_2).
\end{align*}

Let us finally show that closing the ports $(f_1,e_1)$ and $(f_2,e_2)$ corresponds to introducing a Dirac structure $D_2$ on the vector bundle $M\times \R^2\to M$ modeling the space of ports. Consider the (invariant) Dirac structure
\[
D_2=\{(q,v,p,f_1,e_1,f_2,e_2)\st f_1+f_2=0,e_1=e_2\}\subset (M\times \R^2)\oplus (M\times \R^2)^*,
\]
invariant in the sense that it does not depend on the base point. Closing the ports in the \textsc{obio}-structure $A$ gives the \textsc{bio}-structure
\[
A=\left(TM,M\times \R^2,D_1,D_2,p_A\right),
\]
which has equations of motion
\begin{align*}
((x, \dot{x}),{\rm d}E(x)-p_{A}^*(e_1,e_2)) & \in D_1, \\
((f_1,f_2), (e_1,e_2)) & \in D_2,\\
p_A(x, \dot{x})& = (f_1,f_2).
\end{align*}
The second equations tells us that the voltages $e_1$ and $e_2$ are the same at the nodes $N_1$ and $N_2$, and that the currents $f_1$ and $f_2$ are equal. This is depicted in Diagram~\ref{dia:cir1} (right).
\end{example}

\vspace{.3cm}

\begin{example}\label{ex:LCcir2} A similar example is given by a disconnected LC circuit with ports $(f,e)=(f_i,e_i)$, $i=1,\dots,4$, with two components as shown in Diagram~\ref{dia:cir2}.
\begin{figure}[H]
\centering
\includegraphics{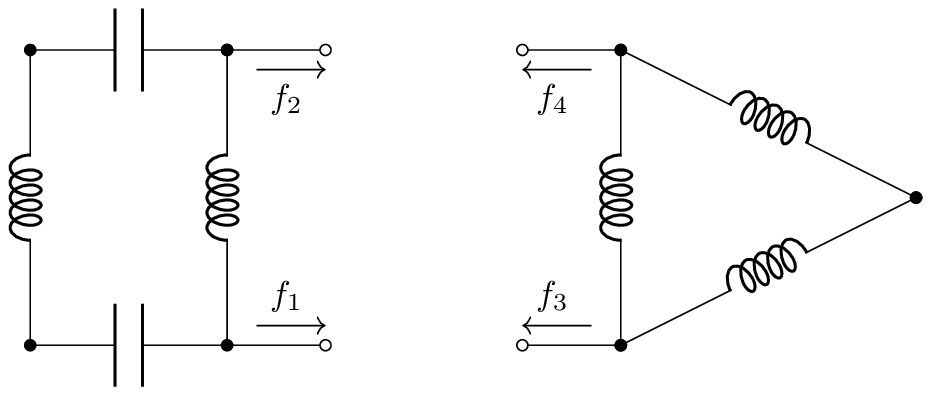}
\caption{The circuit of Example~\ref{ex:LCcir2}}\label{dia:cir2}
\end{figure}

The space of charges is a 7-dimensional vector space. One can proceed like in the previous example, and model it using an \textsc{obio}-structure of the form
\[
A=(TM,M\times\R^4,D_1,(M\times\R^4)\oplus (M\times\R^4)^*,p_A).
\]
Closing the ports with the Dirac structure $D_2\subset (M\times\R^4)\oplus (M\times\R^4)^*)$ given by
\[
D_2=\{(q,v,p,f_1,f_2,f_3,f_4,e_1,e_2,e_3,e_4)\st f_1=-f_3, f_2=-f_4, e_1=e_3,e_2=e_4\}
\]
leads to the following circuit:
\begin{figure}[H]
\centering
\includegraphics{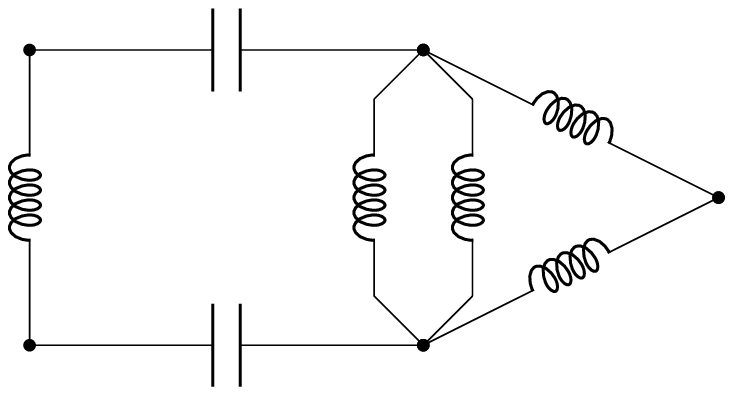}
\end{figure}
\end{example}
\vspace{.3cm}

\begin{example} Consider the system in Diagram~\ref{dia:pen2} (left). It  consists of two subsystems: a pendulum of mass $M$ (with a massless rod of length $\ell=1$) and a free particle of mass $m$, both subject to the gravitational field. The coordinate $\theta$ determines the position of the mass $M$ and the coordinates $x$ and $y$ determine the position of the mass $m$. The configuration space is $Q=S^1\times\R^2$, and the Hamiltonian of the system is:
\[
H=\frac{p_\theta^2}{2M}+\frac{(p_x^2+p_y^2)}{2m} - Mg\cos\theta - mgy,
\]
which is a function on $T^*Q=T^*(S^1\times\R^2)$. We denote by $D_1=\operatorname{graph}(\omega_Q)$ the standard Dirac structure on the tangent bundle $TT^*Q\to T^*Q$. 

\begin{figure}[H]
\centering
\includegraphics{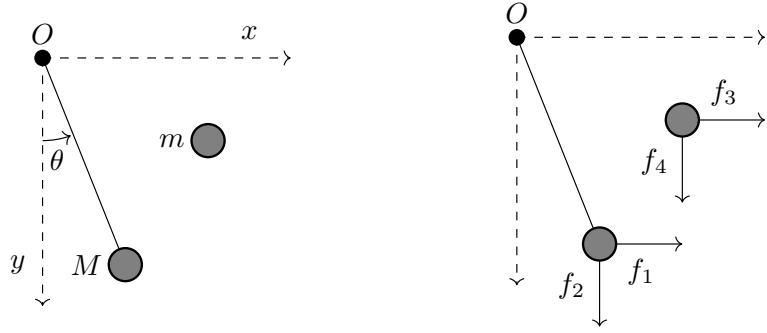}
\caption{Adding ports to a pendulum and a free mass}\label{dia:pen2}
\end{figure}

First, we add ports representing the velocities as shown in Diagram~~\ref{dia:pen2} (right), so that the resulting system is an~\textsc{obio}-structure. This is accomplished using the map $p_A: TT^*Q\to T^*Q\times \R^4$, fibered over the identity on $T^*Q$, given by
\[
p_A(\dot \theta,\dot x,\dot y,\dot p_\theta,\dot p_x,\dot p_y)=(\dot\theta\cos\theta,-\dot\theta\sin\theta,\dot x,\dot y),
\]
where we have omitted the base point in $T^*Q$ for simplicity.

The system with ports is described by the following~\textsc{obio}-structure:
\[
A=(TT^*Q,T^*Q\times\R^4,D_1,(T^*Q\times\R^4)\oplus (T^*Q\times\R^4)^*,p_A).
\]
The dual of the map $p_A$, $(p_A)^*:T^*Q\times (\R^4)^*\to T^*T^*Q$ is given by
\[
(p_A)^*(e_1,e_2,e_3,e_4)=(e_1\cos\theta-e_2\sin\theta,e_3,e_4,0,0,0),
\]
where again we have omitted the base point.

We close the ports with the Dirac structure $D_2\subset (T^*Q\times\R^4)\oplus (T^*Q\times\R^4)^*$ given by:
\[
D_2=\{(f_1,f_2,f_3,f_4)\oplus (e_1,e_2,e_3,e_4)\st f_1=f_3, f_2=f_4, e_1=-e_3, e_2=-e_4\}.
\]
Taking the energy function $E=H$, the dynamics of $A$ reads
\begin{align*}
((\alpha,\dot \alpha),{\rm d}E(\alpha)-(e_1\cos\theta-e_2\sin\theta,e_3,e_4,0,0,0)) & \in D_1, \\
((f_1,f_2,f_3,f_4), (e_1,e_2,e_3,e_4)) & \in D_2,\\
(\dot\theta\cos\theta,-\dot\theta\sin\theta,\dot x,\dot y)& = (f_1,f_2,f_3,f_4),
\end{align*}
with $\alpha=(\theta,x,y,p_\theta,p_x,p_y)\in T^*Q$. The first relation reads
\begin{align*}
\dot p_\theta&=-Mg\sin\theta+e_1\cos\theta - e_2\sin\theta,\qquad \dot p_x=-e_1,\qquad  \dot p_y=mg-e_2 ,\\
 \dot \theta&=p_\theta/M,\qquad  \dot x=p_x/m,\qquad  \dot y=p_y/m,
\end{align*}
which in particular leads to the following relation for the momenta:
\[
\dot p_\theta=-(M+m)g\sin\theta-\dot p_x\cos\theta +\dot p_y\sin\theta.
\]
The second relation implies 
\[
\dot\theta\cos\theta=\dot x, \quad -\dot\theta\sin\theta=\dot y, \quad e_1=-e_3,\quad e_2=-e_4.
\]
In particular, we have $x(t)=\sin(\theta(t))+c_1$ and $y(t)=\cos(\theta(t))+c_2$. Writing $\theta(0)=\theta_0$, $x(0)=x_0$ and $y(0)=y_0$ for the initial condition, we have
\[
 x(t)=\sin(\theta(t))-\sin\theta_0+x_0,\qquad y(t)=\cos(\theta(t))-\cos\theta_0+y_0.
\]
This means that the particle $m$ describes the motion of a pendulum with a massless rod of lenght $1$ pinned at the point $O'=(-\sin\theta_0+x_0,-\cos\theta_0+y_0)$, and which oscillates in phase with the pendulum of mass $M$.  Therefore, when an initial condition $(\theta_0,x_0,y_0)$ is chosen in such a way that $x_0=\sin(\theta_0)$ and $y_0=\cos(\theta_0)$, the masses $m$ and $M$ behave as a (single) pendulum of total mass $m+M$. In other words, for these specific initial conditions, both masses stick together.
\end{example}
\vspace{.3cm}

\begin{remark}\label{rem:1} As the examples suggest, the operation of interconnection of \textsc{ofio}-systems or \textsc{obio}-systems is closely related to that of composition of \textsc{ph}-systems. Actually, it is not hard to check that both operations are essentially the same. More precisely, if two different physical systems are described as either \textsc{ofio}-systems (or \textsc{obio}-systems) or as \textsc{ph}-systems, then their interconnection in the framework of \textsc{ofio}-systems (or \textsc{obio}-systems) coincides with their composition as \textsc{ph}-systems.
\end{remark}

\section{Future work}

This paper provides a new and intrinsic description of port-Hamiltonian systems using backward and forward operations of both Dirac and coistropic structures. We describe here the future research lines this work will lead to:

\begin{enumerate}[a)]
\item Starting from the notions introduced in Sections~\ref{Sec:PHS} and~\ref{sec:InputOutput}, we would like to find a categorical language to describe the relevant definitions and operations in the theory of port-Hamiltonian systems. We believe that the understanding of many constructions in the literature of port-Hamiltonian systems would benefit from such a categorical language. In Appendix~\ref{ap:B} we briefly discuss the equivalence between the forward and backward categories.

    Using a suitable definition of morphisms in the category as maps which preserve the relevant Dirac geometry, one might be able to cast many reduction results (such as those in~\citep{ReductionImplicit}) in a uniform geometric framework . This would include the symmetry reduction of the examples discussed in this paper and the relation with those discussed in~\citep{2012CenRaYo}.

\item The derivation of explicit solutions of Dirac systems is usually very difficult or even impossible and, therefore, it would be interesting to derive ad-hoc numerical methods to tackle this problem. In this sense, our paper clearly uncovers the underlying geometry of these Dirac systems and the interconnection of them from simpler pieces. We intend to study in a future paper the discrete version of the previous construction identifying a suitable notion of discrete Dirac structure, its relation with Morse functions and their interconnection. This discrete version could possible lead us to introduce new geometric integrators (see~\citep{MW_Acta,HLW_book}) discretizing the principles instead of the full differential-algebraic equations. Some steps in this direction have already appeared in~\citep{PaLe_17}.
\item Recently, there has been an increasing interest in modeling engineering and robotic systems which typically involve  a hybrid description of both continuous and discrete dynamics (see e.g.~\citep{vDsScgu_book}). In some cases, this is addressed employing a mixture of logic-based switching and difference/differential equations~\citep{Liberzon_book}. Many systems in engineering and some physical systems can be modeled using such a mathematical framework and it is natural to think that Dirac structures may be useful to model some classes of hybrid systems.

\end{enumerate}

 As a final remark we would like to mention the researh done on infinite dimensional Dirac structures for port-Hamiltonian systems involving partial differential equations in~\cite{vdSMaschkeDistributed} for the interested readers, the discrete couternpart of the so-called Dirac-Stoke structures is studied in~\cite{DiscreteBoundary}.

\appendix

\section{Isotropic, coisotropic and Dirac structures on vector spaces}\label{ap:A}

In this appendix we give a useful description of Dirac, isotropic and coisotropic structures on vector spaces.

Recall that the simplest cases of Dirac structures are those given by a form $\omega$ or a bivector $\Lambda$, namely the Dirac structures $D_\omega$ and $D_\Lambda$ given by the graph of $\omega^\flat\colon V\to V^*$ and $\sharp_\Lambda\colon V^*\to V$ respectively. The cases $\omega=0$ and $\Lambda=0$ yield the Dirac structures $D_\omega=V\oplus \{0\}$ and $D_\Lambda=\{0\}\oplus V^*$. If $\Sigma \subset V\oplus V^*$ is a subspace, we will use the notations $F_\Sigma$ and $F^{(*)}_\Sigma$ for the projections of $\Sigma$ on $V$ and $V^*$. In particular, for a Dirac structure $D$ we write $F_D$ and $F^{(*)}_D$.

Assume that $F\subset V$ is a subspace and $\omega_{F}$ is a 2-form on $F$. We can define the Dirac structure on $V$ given by
\[
D_{V,\omega_{F}}=\{(v,\alpha)\in V\oplus V^*\st v\in F, \; \omega_{F}(v,w)=\alpha(w) \; \mbox{for all } w\in F\},
\]
which should not be confused with $D_{\omega_F}$ which is a Dirac structure on $F$. If we choose a complement $F_1$ of $F$ so that $V=F\oplus F_1$, we represent an element $v\in V$ as $(v_0,v_1)$ and an element $\alpha\in V^*$ as $(\alpha_0,\alpha_1)\in F^*\oplus F_1^*$. The Dirac structure $D_{V,\omega_{F}}$ can then be represented as a direct sum of two Dirac structures on $F$ and $F_1$ as follows
\[
D_{V,\omega_{F}} = D_{\omega_F}\oplus\left(\{0\} \oplus F^*_1 \right)\subset \left(F\oplus F^*\right)\oplus \left(F_1\oplus F_1^*\right),
\]
where $\{0\} \oplus F^*_1$ is the Dirac structure on $F_1$ given by the bivector $\Lambda=0$. Note that for the Dirac structure $D_{V,\omega_{F}}$ one has $F_{D_{V,\omega_{F}}}=F$.

Conversely, let $D$ be a given Dirac structure on $V$. In $F_D$ we can define a presymplectic structure $\omega_{F_D}$ defined by the condition $\omega_{F_D} (v_1, v_2) = \alpha (v_2)$, for all $v_1, v_2 \in F_D$ and all $\alpha$ such that $(v_1, \alpha) \in D$. The proof that $\omega_{F_D}$ is a well defined form only makes use of the isotropy of $D$ which implies that if $\Sigma\subset V\oplus V^*$ is an isotropic subspace, one can also define a 2-form $\omega_{F_\Sigma}$ on $F_\Sigma$. Then applying the previous construction with $F = F_D$ and $\omega_F = \omega_{F_D}$ we recover the Dirac structure $D$, that is, $D_{V, \omega_{F_D}} = D$.

Using the representation $D_{V, \omega_F}$ of a given Dirac structure on $V$ one can describe the family of all isotropic structures $\Sigma$ on $V$ such that $F_{\Sigma} = F$. In fact, it is easy to check that $\Sigma$  must be of the form $\Sigma = D_{\omega_F} \oplus \{0\}\oplus F^{\circ_{F_1}}_2$
where $F_2$ represents an arbitrary subspace of $F_1$ and $F^{\circ_{F_1}}_2$ denotes the annihilator of $F_2$ in $F_1$, i.e. $F^{\circ_{F_1}}_2=\{\alpha\in F^*_1\st \alpha(v)=0,\, \mbox{for all } v\in F_2\}$. Therefore $F^{\circ_{F_1}}_2$ represents an arbitrary subspace of $F_1^*$. The latter subspace is maximal if $F_2 = \{0\}$ which gives $\Sigma = D_{V, \omega_F}$ and it is minimal if $F_2 = F_1$, which gives $\Sigma = D_{\omega_F} \oplus \{0\}\oplus \{0\}$.

Any coisotropic structure on $V$ can be described as the orthogonal complement of an isotropic structure $\Sigma = D_{\omega_F} \oplus \{0\}\oplus F^{\circ_{F_1}}_2$ as described above, which gives $\Sigma^\perp = D_{\omega_F} \oplus F_2 \oplus F_1^*$. The latter subspace is maximal if $F_2 = F_1$ which gives $\Sigma^\perp = D_{\omega_F} \oplus F_1\oplus F_1^*$
and it is minimal if $F_2 = \{0\}$ which gives $\Sigma^\perp = D_{\omega_F} \oplus \{0\}\oplus F_1^*=D_{V,\omega_F}$. Note that $F_{(\Sigma^\perp)} = F_{\Sigma} \oplus F_2$ which contains $F_{\Sigma}$ and it is equal to it if and only if $F_2 = \{0\}$ which, in turn, happens if and only if $\Sigma^\perp = \Sigma$, that is,
$\Sigma$ is a Dirac structure.

Furthermore, $F_2 \oplus F_1^\ast$ can be decomposed as
$F_2 \oplus F_1^\ast = F_2 \oplus F_2^\ast \oplus F_3^\ast$ where $F_3 \subset F_1$ is any subspace of $F_1$ such that $F_2 \oplus F_3 = F_1$, and we can conclude that $\Sigma^\perp$ can be decomposed as the direct sum of three structures corresponding to the decomposition
$V = F\oplus F_3 \oplus F_2$, namely
$\Sigma^\perp = D_{\omega_F} \oplus (\{0\}\oplus F_3^\ast) \oplus F_2 \oplus F_2^\ast$.
Note that $D_{F\oplus F_3, \omega_F} := D_{\omega_F} \oplus (\{0\}\oplus F_3^\ast)$ is a Dirac structure on
$F \oplus F_3$ and $F_2 \oplus F_2^\ast$ is the maximal coisotropic structure on $F_2$.

\section{Relation between the forward and backward categories}\label{ap:B}

There is a close relation between the forward and backward of Dirac structures which can be described in a precise way in the language of categories. We review the main ingredients for the case of vector spaces following~\citep{2012CenRaYo}. We refer the reader to that reference for a complete discussion.

A \emph{Dirac space} is a pair $(U,D_U)$, where $U$ is a vector space and $D_U$ is a Dirac structure on $U$. First, we define the category \texttt{Forw-DS}. Objects are Dirac spaces, and a morphism
\[
\varphi^{\F}\colon (U,D_U)\to (V,D_V)
\]
is a linear map $\varphi\colon U\to V$ such that $\F \varphi(D_U)=D_V$. The composition of morphisms is such that $\varphi^{\F}\circ \psi^{\F}=(\varphi\circ \psi)^{\F}$. For each Dirac space $(V,D_V)$ the identity morphism is $(1_{V})^\F$, where $1_V\colon V\to V$ is the identity. The category \texttt{Back-DS} is defined analogously: objects are \texttt{DS}, and a morphism
\[
\varphi^{\B}\colon (U,D_U)\to (V,D_V)
\]
is a map $\varphi\colon U\to V$ such that $\B \varphi(D_V)=D_U$. The composition of morphisms is such that $\varphi^{\B}\circ
\psi^{\B}=(\varphi\circ \psi)^{\B}$. For each Dirac space $(V,D_V)$ the identity morphism is $(1_{V})^\B$.

We will show that the categories \texttt{Forw-DS} and \texttt{Back-DS} are isomorphic under some natural identifications. More precisely, we assume the identification $U^{**}\equiv U$ for finite dimensional vector spaces and $f^{**}\equiv f$ for linear maps $f\colon U\to V$ between vector spaces. If $D_U\subset U\oplus U^*$ is a Dirac structure, the \emph{twist of $D$}, $\boldsymbol{t}(D_U)$, is the following Dirac structure on $U^*$:
\begin{equation*}\label{eq:twist}
\boldsymbol{t}(D_U) = \{\alpha\oplus u\in U^*\oplus U \st u\oplus\alpha \in D_U\}\subset U^*\oplus U.
\end{equation*}
The twist satisfies
\[
\boldsymbol{t}\left(\boldsymbol{t}(D_U)\right)=D_U.
\]

Using the twist, we define a contravariant functor $\delta\colon \texttt{Forw-DS}\to \texttt{Back-DS}$ as follows. For an object $(U,D_U)$ in $\texttt{Forw-DS}$ we set $\delta\left((U,D_U)\right)=(U^*,\boldsymbol{t}(D_U))$. For a morphism $f^\F$, we set $\delta(f^\F)=(f^*)^\B$, where $f^*$ is the dual of $f$. Likewise, we define a contravariant functor $\bar\delta\colon \texttt{Back-DS}\to \texttt{Forw-DS}$ defined as follows. For an object $(U,D_U)$ in $\texttt{Back-DS}$, we set $\bar\delta\left((U,D_U)\right)=(U^*,\boldsymbol{t}(D_U))$ and for a morphism $f^\B$, we set $\bar\delta(f^\B)=(f^*)^\F$, where $f^*$ is the dual of $f$. The fundamental result is the following:
\begin{proposition} The contravariant functors $\delta$ and $\bar\delta$ satisfy
\[
\bar \delta\circ\delta= 1_{\texttt{Forw-DS}},\qquad \delta\circ\bar\delta= 1_{\texttt{Back-DS}}.
\]
\end{proposition}
\noindent Thus, under the identifications we are assuming between a vector space and its bidual, the categories \texttt{Forw-DS} and \texttt{Back-DS} are dually isomorphic.
\vspace*{.3cm}

This equivalence applies for instance to the case of \textsc{fio} and \textsc{bio}-structures as follows. Given a \textsc{fio}-structure
\[
A=\big(\pi_{(U_1,M)},\pi_{(U_2,M)},D_{U_1},D_{U_2},g_{U_2U_1}\big),
\]
the construction of $\Phi_A$ from $g_{U_2U_1}\colon U_2\to U_1$ gives rise to a dual construction, namely $g^*_{U_2U_1}\colon U_1^*\to U_2^*$ and $\Phi_A^*\colon U_1^*\to U_1^*\oplus U_2^*$. One can check that
\[
\Phi_A^*(u_1^*)=u_1^*+g^*_{U_2U_1}(u_1^*).
\]
We observe that this dual construction has the same formal property of the map $\Psi_B$ where $B$ is the following \textsc{bio}-structure:
\[
B=\big(\pi_{(U^*_1,M)},\pi_{(U^*_2,M)},\boldsymbol{t}(D_{U_1}),\boldsymbol{t}(D_{U_2}),p_{U_1U_2}=g^*_{U_2U_1}\big).
\]

\bibliography{References}

\end{document}